\documentclass[11pt]{article}

\usepackage[letterpaper,margin=1in]{geometry}
\usepackage{amsmath, amssymb, amsthm, amsfonts}
\usepackage{bbm}
\usepackage{bm}

\usepackage{subcaption}

\usepackage{ifthen}
\usepackage{comment}
\usepackage{tikz}
\usetikzlibrary{positioning,decorations.pathreplacing}
\usepackage{comment}

\newcommand{\emphbf}[1]{{\bfseries #1}}

\usepackage{cite}
\usepackage{appendix}
\usepackage{graphicx}
\usepackage{color}
\usepackage{algorithm}
\usepackage[noend]{algorithmic}
\usepackage{epstopdf}
\usepackage{wrapfig}
\usepackage{paralist}
\usepackage[textsize=tiny]{todonotes}

\usepackage{framed}
\usepackage[framemethod=tikz]{mdframed}
\usepackage[bottom]{footmisc}
\usepackage{enumitem}
\setitemize{noitemsep,topsep=3pt,parsep=3pt,partopsep=3pt}
\usepackage[font=small]{caption}
\usepackage{xspace}

\newtheorem{theorem}{Theorem}[section]
\newtheorem{lemma}[theorem]{Lemma}
\newtheorem{meta-theorem}[theorem]{Meta-Theorem}

\newtheorem{corollary}[theorem]{Corollary}

\newtheorem{definition}[theorem]{Definition}

\definecolor{darkgreen}{rgb}{0,0.5,0}
\usepackage{hyperref}
\hypersetup{
    unicode=false,          
    colorlinks=true,        
    linkcolor=red,          
    citecolor=darkgreen,        
    filecolor=magenta,      
    urlcolor=cyan           
}
\usepackage[capitalize]{cleveref}

\newcommand{\eps}{\varepsilon}

\newcommand{\ceil}[1]{\lceil #1 \rceil}

\newcommand{\acomment}[1]{{\color{red}\ttfamily A: {#1}}}

\newcommand{\important}[1]{\textbf{#1}}

\newcommand{\exclude}[1]{}

\newcommand{\FullOrShort}{full}

\ifthenelse{\equal{\FullOrShort}{full}}{

  \newcommand{\fullOnly}[1]{#1}
	\newcommand{\tempfullOnly}[1]{#1}
  \newcommand{\shortOnly}[1]{}
  
  }{
    \usepackage{times}
    \usepackage[compact]{titlesec}
    
    \newcommand{\fullOnly}[1]{}
		\newcommand{\tempfullOnly}[1]{}
		\newcommand{\shortOnly}[1]{#1}
    \newcommand{\IncludePictures}[1]{}
   
  }

\begin{document}

\date{}

\title{Synchronization Strings: Explicit Constructions, Local Decoding, and Applications\footnote{Supported in part by the National Science Foundation through grants CCF-1527110 and CCF-1618280.\shortOnly{ An extended version of this paper can be found at~\cite{haeupler2017synchronization3:ARXIV}.}}.}

\author{Bernhard Haeupler\\Carnegie Mellon University\\ \texttt{haeupler@cs.cmu.edu} \and
Amirbehshad Shahrasbi\\Carnegie Mellon University\\ \texttt{shahrasbi@cs.cmu.edu}}

\maketitle
	
\begin{abstract}

This paper gives new results for \emphbf{synchronization strings}, a powerful combinatorial object that allows to efficiently deal with insertions and deletions in various communication problems:
\begin{itemize}
				\item We give a \emphbf{deterministic, linear time synchronization string construction}, improving over an $O(n^5)$ time randomized construction. Independently of this work, a deterministic $O(n \log^2 \log n)$ time construction was just put on arXiv by Cheng, Li, and Wu.
				\item We give a \emphbf{deterministic construction of an infinite synchronization string} which outputs the first $n$ symbols in $O(n)$ time. Previously it was not known whether such a string was computable.
				\item Both synchronization string constructions are \emphbf{highly explicit}, i.e., the $i^{th}$ symbol  can be deterministically computed in $O(\log i)$ time. 
				\item This paper also introduces a generalized notion we call \emphbf{long-distance synchronization strings}. Such strings \emphbf{allow for local and very fast decoding}. In particular only $O(\log^3 n)$ time and access to logarithmically many symbols is required to decode any index. 
		\end{itemize}

\medskip
\noindent The paper also provides several applications for these improved synchronization strings:
		\begin{itemize}
			\item For any $\delta < 1$ and $\eps > 0$ we provide an \emphbf{insdel error correcting block code} with rate $1 - \delta - \eps$ which can correct any $O(\delta)$ fraction of insertion and deletion errors in $O(n \log^3 n)$ time. This \emphbf{near linear computational efficiency} is surprising given that we do not even know how to compute the (edit) distance between the decoding input and output in sub-quadratic time.
			\item We show that local decodability implies that error correcting codes constructed with long-distance synchronization strings can not only efficiently recover from $\delta$ fraction of insdel errors but, similar to [Schulman, Zuckerman; TransInf'99], also from any $O(\delta / \log n)$ fraction of \emphbf{block transpositions and block replications}. These block corruptions allow arbitrarily long substrings to be swapped or replicated anywhere.
			\item We show that highly explicitness and local decoding allow for \emphbf{infinite channel simulations with   exponentially smaller memory and decoding time requirements}. These simulations can then be used to give the first \emphbf{near linear time interactive coding scheme for insdel errors}, similar to the result of [Brakerski, Naor; SODA'13] for Hamming errors. 
			\end{itemize}

\end{abstract}
	
\setcounter{page}{0}
\thispagestyle{empty}

\newpage
\section{Introduction}

This paper gives new results for \emphbf{$\eps$-synchronization strings}, a powerful combinatorial object that can be used to effectively deal with insertions and deletions in various communication problems.

\medskip

Synchronization strings are pseudo-random non-self-similar sequences of symbols over some finite alphabet that can be used to index a finite or infinte sequence of elements similar to the trivial indexing sequence $1,2,3,4,\ldots,n$. In particular, if one first indexes a sequence of $n$ elements with the trivial indexing sequence and then applies some $k$ insertions or deletions of indexed elements one can still easily recover the original sequence of elements up to $k$ half-errors, i.e., erasures or substitutions (where substitutions count twice). An $\eps$-synchronization strings allows essentially the same up to an arbitrarily small error of $\eps n$ half-errors but instead of having indexing symbols from a large alphabet of size $n$, which grows with the length of the sequence, a finite alphabet size of $\eps^{-O(1)}$ suffices for $\eps$-synchronization strings. Often this allows to efficiently transform insertion and deletion errors into ordinary Hamming errors which are much better understood and easier to handle. 

\medskip

One powerful application of synchronization strings is the design of efficient insdel error correcting codes (ECC), i.e., codes that can efficiently correct insertions and deletions. While codes for Hamming errors have been well understood making progress on insdel codes has been difficult~\cite{golomb1963synchronization,Levenshtein65,mercier2010survey,sloane2002single,GL-isit16,GW-random15}. Synchronization strings solve this problem by transforming any regular error correcting block code $C$ with a sufficiently large finite alphabet into an essentially equally efficient insdel code by simply indexing the symbols of $C$.
This leads to the first insdel codes that approach the Singleton bound, i.e., for any $\delta <1 $ and $\eps >0$ one can get an insdel code with rate $1 - \delta - \eps$ which, in quadratic time, recovers from any $\delta$ fraction of insertions or deletions. Further applications are given in \cite{haeupler2017synchronization2:ARXIV}. Most importantly, \cite{haeupler2017synchronization2:ARXIV} introduces the notion of a channel simulation which allows one to use any insertion deletion channel like a black-box regular symbol corruption channel with an slightly increased error rate. This can be used to give the first computationally efficient interactive coding schemes for insdel errors and the first interactive coding scheme for insdel errors whose communication rate goes to one as the amount of noise goes to zero. 

\medskip

This paper provides drastically improved constructions of finite and infinite synchronization strings and a stronger synchronization string property which allows for decoding algorithms that are local and significantly faster. We furthermore give several applications for these results, including near linear time insertion-deletion codes, a near linear time coding scheme for interactive communication over insertion-deletion channels, exponentially better channel simulations in terms of time and memory, infinite channel simulations, and codes that can correct block transposition and block replication corruptions.

\section{Our Results, Structure of this Paper, and Related Work}

Next we give an overview over the main results and the overall structure of this paper. We also put our result in relation to related prior works.

\subsection{Deterministic, Linear Time, Highly Explicit Construction of Infinte Synchronization Strings} 

In \cite{haeupler2017synchronization} the authors introduced synchronization strings and gave a $O(n^5)$ time randomized synchronization string construction. This construction could not be easily derandomized. In order to provide deterministic explicit constructions of insertion deletion block codes \cite{haeupler2017synchronization} introduced a strictly weaker notion called self-matching strings, showed that these strings could be used for code constructions as well, and gave a deterministic $n^{O(1)}$ time self-matching string construction. Obtaining a deterministic synchronization string construction of was left open. \cite{haeupler2017synchronization} also showed the existence of infinite synchronization strings. This existence proof however is highly non-constructive. In fact even the existence of a computable infinite synchronization string was left open; i.e., up to this paper there was no algorithm that would compute the $i^{th}$ symbol of some infinite synchronization string in finite time.
	
\medskip	
	
In this paper we give deterministic constructions of finite and infinite synchronization strings. Instead of going to a weaker notion, as done in \cite{haeupler2017synchronization}, Section~\ref{sec:LongDistSync} introduces a stronger notion called long-distance synchronization strings. Interestingly, while the existence of these generalized synchronization strings can be shown with a similar Lovasz local lemma based proof as for plain synchronization strings, this proof allows for an easier derandomization, which leads to a \emphbf{deterministic polynomial time construction of (long-distance) synchronization strings}. Beyond this derandomization the notion of long-distance synchronization strings turns out to be very useful and interesting in its own right, as will be shown later. 

\medskip	
	
Next, two different boosting procedures, which make synchronization string constructions faster and more explicit, are given. The first boosting procedure, given in Section~\ref{sec:BoostingStepI}, leads to a \emphbf{deterministic linear time synchronization string construction}. We remark that concurrently and independently Cheng, Li, and Wu obtained a deterministic $O(n \log^2 \log n)$ time synchronization string construction~\cite{cheng2017synchronization}. 

\medskip	

Our second boosting step, which is introduced in Section~\ref{sec:BoostingStepII}, makes our synchronization string construction \emphbf{highly-explicit}, i.e., allows to compute any position of an $n$ long synchronization string in time $O\left(\log n\right)$. This highly-explicitness is an property of crucial importance in most of our new applications. 

\medskip	

Lastly, in Section~\ref{sec:infiniteConstruction} we give a simple transformation which allows to use any construction for finite length synchronization strings and utilize it to give an construction of an \emphbf{infinite synchronization string}. This transformation preserves highly-explicitness. Infinite synchronization strings are important for applications in which one has no a priori bound on the running time of a system, such as, streaming codes, channel simulations, and some interactive coding schemes. Overall we get the following simple to state theorem:

\begin{theorem}
For any $\eps > 0$ there exists an infinite $\eps$-synchronization string $S$ over an alphabet of size $\eps^{-O(1)}$ and a deterministic algorithm which for any $i$ takes $O(\log i)$ time to compute $S[i,i+\log i]$, i.e., the $i^{th}$ symbol of $S$ (as well as the next $\log i$ symbols). 
\end{theorem} 

Since any substring of an $\eps$-synchronization string is also an $\eps$-synchronization string itself this infinite synchronization string construction also implies a deterministic linear time construction of finite synchronization strings which is fully parallelizable. In particular, for any $n$ there is a linear work parallel $\mathbf{NC}^1$ algorithm with depth $O(\log n)$ and $O(n / \log n)$ processors which computes the $\eps$-synchronization string $S[1,n]$.


%

\subsection{Long Distance Synchronization Strings and Fast Local Decoding}

Section~\ref{sec:localDecoding} shows that the long-distance property we introduced in Section~\ref{sec:LongDistSync}, together with our highly explicit constructions from Section~\ref{sec:BoostingStepII}, allows the design of a much faster and highly local decoding procedure. In particular, to decode the index of an element in a stream that was indexed with a synchronization string it suffices to look at the only the $O(\log n)$ previously received symbols. The decoding of the index itself furthermore takes only $O(\log^3 n)$ time and can be done in a streaming fashion. This is significantly faster than the $O(n^3)$ streaming decoder or the $O(n^2)$ global decoder given in \cite{haeupler2017synchronization}.

\bigskip

The paper furthermore gives several applications which demonstrate the power of these improved synchronization string constructions and the local decoding procedure. 

\subsection{Application: Codes Against Insdels, Block Transpositions and Replications\shortOnly{ (Appendix~\ref{sec:AppilcationsCode})}}
\subsubsection{Near Linear Time Decodable Error Correcting Codes}
	
Fast encoding and decoding procedures for error correcting codes have been important and influencial in both theory and practice. For regular error correcting block codes the celebrated expander code framework given by Sipser and Spielman~\cite{sipser1996expander} and in Spielman's thesis~\cite{spielman1995computationally} as well as later refinements by Alon, Edmonds, and Luby~\cite{alon1995linear} as well as Guruswami and Indyk~\cite{guruswami2005linear, guruswami2001expander} gave good ECCs with linear time encoding and decoding procedures. Very recently a beautiful work by Hemenway, Ron-Zewi, and Wooters~\cite{hemenway2017local} achieved linear time decoding also for capacity achieving list decodable and locally list recoverable codes. 

\medskip

The synchronization string based insdel codes in \cite{haeupler2017synchronization} have linear encoding times but quadratic decoding times. As pointed out in \cite{haeupler2017synchronization} the later seemed almost inherently to the harsher setting of insdel errors because \emph{``in contrast
to Hamming codes, even computing the distance between the received and the sent/decoded string
is  an  edit  distance  computation. Edit  distance  computations  in  general  do  usually  not  run  in
sub-quadratic time,  which is not surprising given the recent SETH-conditional lower bounds~\cite{backurs2015edit}''}. Very surprisingly to us, our fast decoding procedure allows us to construct insdel codes with near linear decoding complexity:

\begin{theorem}
For any $\delta < 1$ and $\eps > 0$ there exists an \emphbf{insdel error correcting block code} with rate $1 - \delta - \eps$ that can correct from any $O(\delta)$ fraction of insertions and deletions in $O(n \log^3 n)$ time. The encoding time is linear and the alphabet bit size is near linear in $\frac{1}{\delta+\eps}$.
\end{theorem}

Note that for any input string the decoder finds the codeword that is closest to it in edit distance, if a codeword with edit distance of at most $O(\delta n)$ exists. However, computing the distance between the input string and the codeword output by the decoder is an edit distance computation. Shockingly, even now, we do not know of any sub-quadratic algorithm that can compute or even crudely approximate this distance between input and output of our decoder, even though intuitively this seems to be much easier almost prerequisite step for the distance minimizing decoding problem itself. After all, decoding asks to find the closest (or a close) codeword to the input from an exponentially large set of codewords, which seems hard to do if one cannot even approximate the distance between the input and any particular codeword.

\subsubsection{Application: High-Rate InsDel Codes that Efficiently Correct Block Transpositions and Replications}

\fullOnly{Section}\shortOnly{Appendix}~\ref{sec:BlockTranspositionAndReplication} gives another interesting application of our local decoding procedure. In particular, we show that 
local decodability directly implies that insdel ECCs constructed with our highly-explicit long-distance synchronization strings can not just efficiently recover from $\delta$ fraction of insdel errors but also from any $O(\delta / \log n)$ fraction of \emphbf{block transpositions and block replications}. Block transpositions allow for arbitrarily long substrings to be swapped while a block replication allows for an arbitrarily long substring to be duplicated and inserted anywhere else. A similar result, albeit for block transpositions only, was shown by Schulman, Zuckerman \cite{schulman1999asymptotically} for the efficient constant distance constant rate insdel codes given by them. They also show that the $O(\delta / \log n)$ resilience against block errors is optimal up to constants. 

\subsection{Application: Exponentially More Efficient Infinite Channel Simulations\shortOnly{ (Appendix~\ref{sec:InfiniteSimulation})}}

\cite{haeupler2017synchronization2:ARXIV} introduced the powerful notion of a channel simulation. In particular, 
\cite{haeupler2017synchronization2:ARXIV} showed that for any adversarial one-way or two-way insdel channel one can put a simple black-box at both ends such that to any two parties interacting with these black-boxes the behavior is indistinguishable from a much nicer Hamming channel which only introduces (a slightly larger fraction of) erasures and symbol corruptions. To achieve this these black-boxes were required to know a prior for how many steps $T$ the channel would be used and required an amount of memory size that is linear in $T$. Furthermore, for each transmission at a time step $t$ the receiving black-box would perform a $O(t^3)$ time computation. We show that using our locally decodable highly explicit long-distance synchronization strings can reduce both the memory requirement and the computation complexity exponentially. In particular each box is only required to have $O(\log t)$ bits of memory (which is optimal because at the very least it needs to store the current time) and any computation can be done in $O(\log^3 t)$ rounds. Furthermore due to our infinite synchronization string constructions the channel simulations black-boxes are not required to know anymore for how much time overall the channel will be used. These drastic improvements make channel simulations significantly more useful and indeed potentially quite practical. 

\subsection{Application: Near-Linear Time Interactive Coding Schemes for InsDel Errors\shortOnly{ (Appendix~\ref{sec:NearLinearInteractiveScheme})}}
Interactive coding schemes, as introduced by Schulman~\cite{schulman1992communication, schulman1996Interactive}, allow to add redundancy to any interactive protocol between two parties in such a way that the resulting protocol becomes robust to noise in the communication. Interactive coding schemes that are robust to symbol corruptions have been intensely studied over the last few years~\cite{braverman2014toward,franklin2015optimal, kol2013interactive,haeupler2014interactive:FOCS, brakerski2012efficient,brakerski2013fast,ghaffari2014optimal, gelles2014efficient}. Similar to error correcting codes the main parameters for an interactive coding scheme is the fraction of errors it can tolerate\cite{schulman1992communication, schulman1996Interactive,braverman2014toward,franklin2015optimal} its communication rate\cite{kol2013interactive,haeupler2014interactive:FOCS} and its computational efficiency~\cite{brakerski2012efficient,brakerski2013fast,ghaffari2014optimal, gelles2014efficient}. In particular, Brakerski and Kalai~\cite{brakerski2012efficient} gave the first computationally efficient polynomial time interactive coding scheme. Brakerski and Naor~\cite{brakerski2013fast} improved the complexity to near linear. Lastly, Ghaffari and Haeupler~\cite{ghaffari2014optimal} gave a near-linear time interactive coding scheme that also achieved the optimal maximal robustness. More recently interactive coding schemes that are robust to insertions and deletions have been introduced by Braverman, Gelles, Mao, and Ostrovsky~\cite{braverman2017coding} subsequently Sherstov and Wu~\cite{sherstov2017optimal} gave a scheme with optimal error tolerance and Haeupler, Shahrasbi, and Vitercik~\cite{haeupler2017synchronization2:ARXIV} used channel simulations to give the first computationally efficient polynomial time interactive coding scheme for insdel errors. Our improved channel simulation can be used together with the coding scheme from \cite{ghaffari2014optimal} to directly get the first interactive coding scheme for insertions and deletions with a near linear time complexity - i.e., the equivalent of the result of Brakerski and Naor~\cite{brakerski2013fast} but for insertions and deletions. 

\section{Definitions and Preliminaries}\label{sec:def}

In this section, we provide the notation and definitions we will use throughout the rest of the paper. \fullOnly{We also briefly review key definitions and techniques from \cite{haeupler2017synchronization, haeupler2017synchronization2:ARXIV}.

\subsection{String Notation}}

\shortOnly{\smallskip}

\noindent \textbf{String Notation.} For two strings $S \in \Sigma^n$ and $S' \in \Sigma^{n'}$ be two strings over alphabet $\Sigma$. We define $S \cdot S' \in \Sigma^{n+n'}$ to be their concatenation. For any positive integer $k$ we define $S^k$ to equal $k$ copies of $S$ concatenated together. For $i,j \in \{1, \dots, n\}$, we denote the substring of $S$ from the $i^{th}$ index through and including the $j^{th}$ index as $S[i,j]$. Such a consecutive substring is also called a \emph{factor} of $S$. For $i < 1$ we define $S[i,j] = \bot^{-i+1} \cdot S[1,j]$ where $\bot$ is a special symbol not contained in $\Sigma$. We refer to the substring from the $i^{th}$ index through, but not including, the $j^{th}$ index as $S[i,j)$. The substrings $S(i,j]$ and $S[i,j]$ are similarly defined. $S[i]$ denotes the $i^{th}$ symbol of $S$ and $|S| = n$ is the length of $S$. Occasionally, the alphabets we use are the cross-product of several alphabets, i.e. $\Sigma = \Sigma_1 \times \cdots \times \Sigma_n$. If $T$ is a string over $\Sigma,$ then we write $T[i] = \left[a_1, \dots, a_n\right]$, where $a_i \in \Sigma_i$. Finally, symbol by symbol concatenation of two strings $S$ and $T$ of similar length is $\left[(S_1, T_1), (S_2, T_2), \cdots\right]$.

\smallskip

\noindent \textbf{Edit Distance.} Throughout this work, we rely on the well-known \emph{edit distance} metric defined as follows.
\begin{definition}[Edit distance]
The \emph{edit distance} $ED(c,c')$ between two strings $c,c' \in \Sigma^*$ is the minimum number of insertions and deletions required to transform $c$ into $c'$.
\end{definition}
It is easy to see that edit distance is a metric on any set of strings and in particular is symmetric and satisfies the triangle inequality property. Furthermore, $ED\left(c,c'\right) = |c| + |c'| - 2\cdot LCS\left(c,c'\right)$, where $LCS\left(c,c'\right)$ is the longest common substring of $c$ and $c'$.
\shortOnly{
\begin{definition}[$\eps$-Synchronization String]\label{def:synCode}
String $S \in \Sigma^n$ is an $\eps$-synchronization string if for every $1 \leq i < j < k \leq n + 1$ we have that $ED\left(S[i, j),S[j, k)\right) > (1-\eps) (k-i)$. We call the set of prefixes of such a string an $\eps$-synchronization code. 
\end{definition}
We provide a brief review of main concepts and techniques related to synchronization strings from~\cite{haeupler2017synchronization} which will be used throughout the rest of this paper in Appendix~\ref{App:Prelims}.
}
\global\def\PrelimsSynchStrings{
\smallskip
\shortOnly{
\section{Definitions and Preliminaries}\label{App:Prelims}
\subsection{Relative Suffix Distance}}
\begin{definition}[Relative Suffix Distance]
For any two strings $S, S' \in \Sigma^*$ we define their relative suffix distance $RSD$ as follows:
$$RSD(S,S') = \max_{k > 0} \frac{ED\left(S(|S|-k,|S|],S'(|S'|-k,|S'|]\right)}{2k}$$
\end{definition}

\begin{lemma}\label{lem:RSDMetricProperties}
For any strings $S_1,S_2,S_3$ we have 	
\begin{itemize}
	\item {\bfseries Symmetry:} $RSD(S_1,S_2) = RSD(S_2,S_1)$,
	\item {\bfseries Non-Negativity and Normalization:} $0 \leq RSD(S_1,S_2) \leq 1$,
	\item {\bfseries Identity of Indiscernibles:} $RSD(S_1,S_2) = 0 \Leftrightarrow S_1 = S_2$, and
	\item {\bfseries Triangle Inequality:} $RSD(S_1,S_3) \leq RSD(S_1,S_2) + RSD(S_2,S_3)$.
\end{itemize}
In particular, RSD defines a metric on any set of strings. 
\end{lemma}
\global\def\PrelimsECC{\color{red}
\subsection{Error Correcting Codes}

Next, we give a quick summary of the standard definitions and formalism around error correcting codes. This is mainly for completeness and we remark that readers already familiar with basic notions of error correcting codes might want to skip this part.

\paragraph{Codes, Distance, Rate, and Half-Errors.}

\newcommand{\Sigmain}{{\Sigma'}}
\newcommand{\Sigmaout}{{\Sigma}}

An \emph{error correcting code} $C$ is an injective function which
takes an input string $s \in (\Sigmain)^{n'}$ over alphabet $\Sigmain$
of length $n'$ and generates a \emph{codeword} $C(s) \in \Sigmaout^n$
of length $n$ over alphabet $\Sigmaout$. The length $n$ of a codeword
is also called the \emph{block length}. The two most important
parameters of a code are its distance $\Delta$ and its rate
$R$. The \emph{rate} $R = \frac{n \log |\Sigmaout|}{n' \log
|\Sigmain|}$ measures what fraction of bits in the codewords produced
by $C$ carries non-redundant information about the
input. The \emph{code distance} $\Delta(C)
= \min_{s,s'} \Delta(C(s),C(s'))$ is simply the minimum Hamming
distance between any two codewords. The \emph{relative distance}
$\delta(C) = \frac{\Delta(C)}{n}$ measures what fraction of output
symbols need to be corrupted to transform one codeword into another.

It is easy to see that if a sender sends out a codeword $C(s)$ of code $C$ with relative distance $\delta$ a receiver can uniquely  recover $s$ if she receives a codeword in which less than a $\delta$ fraction of symbols are affected by an \emph{erasure}, i.e., replaced by a special $''?''$ symbol.
Similarly, a receiver can uniquely recover the input $s$ if less than
$\delta / 2$ \emph{symbol corruptions}, in which a symbol is replaced
by any other symbol from $\Sigmaout$, occurred. More generally it is
easy to see that a receiver can recover from any combination of
$k_{e}$ erasures and $k_{c}$ corruptions as long as $k_{e} + 2 k_{c}
< \delta n$. This motivates defining half-errors to incorporate both
erasures and symbol corruptions where an erasure is counted as a
single half-error and a symbol corruption is counted as two
half-errors. In summary, any code of distance $\delta$ can tolerate
any error pattern of less than $\delta n$ half-errors.

We remark that in addition to studying codes with decoding guarantees
for worst-case error pattern as above one can also look at more benign
error models which assume a distribution over error patterns, such as
errors occurring independently at random. In such a setting one looks
for codes which allow unique recovery for typical error patterns,
i.e., one wants to recover the input with probability tending to $1$
rapidly as the block length $n$ grows. While synchronization strings 
might have applications for such codes as well, this paper focuses
exclusively on codes with good distance guarantees which tolerate an
arbitrary (worst-case) error pattern.

\paragraph{Synchronization Errors.}
In addition to half-errors, we study \emph{synchronization errors}
which consist of \emph{deletions}, that is, a symbol being removed
without replacement, and \emph{insertions}, where a new symbol from
$\Sigmaout$ is added anywhere. It is clear
that \important{synchronization errors are strictly more general and
harsh than half-errors}. The above formalism of codes, rate, and distance works equally well for synchronization errors if one replaces the Hamming distance with edit distance. Instead of measuring the number of symbol corruptions required to transform one string into another, \emph{edit distance}
measures the minimum number of insertions and deletions to do so. An insertion-deletion error correcting code, or \emph{insdel code} for short, of relative distance $\delta$ is a set of codewords for which at least $\delta n$ insertions and deletions are needed to transformed any codeword into another. Such a code can correct any combination of less than $\delta n/2$ insertions and deletions. We remark that it is possible for two codewords of length $n$ to have edit distance up to $2n$ putting the (minimum) relative edit distance between zero and two and allowing for constant rate codes which can tolerate $(1 - \eps)n$ insdel errors. 

\paragraph{Efficient Codes}

In addition to codes with a good minimum distance, one furthermore
wants efficient algorithms for the encoding and error-correction tasks
associated with the code. Throughout this paper we say a code is
efficient if it has encoding and decoding algorithms running in time
polynomial in the block length. While it is often not hard to show
that random codes exhibit a good rate and distance, designing codes
which can be decoded efficiently is much harder. We remark that most
codes which can efficiently correct for symbol corruptions are also
efficient for half-errors. For insdel codes the situation is slightly
different. While it remains true that any code that can uniquely be
decoded from any $\delta(C)$ fraction of deletions can also be decoded
from the same fraction of insertions and
deletions~\cite{Levenshtein65} doing so efficiently is often much
easier for the deletion-only setting than the fully general insdel
setting. 
}
\subsection{Synchronization Strings}
\exclude{\color{red}

\begin{lemma}\label{lem:codewordRSDdistance}
If $S$ is an $\eps$-synchronization string, then $RSD(S[1,i],S[1,j]) > 1-\eps$ for any $i < j$, i.e.,  any two codewords associated with $S$ have RSD distance of at least $1-\eps$. 
\end{lemma}
\begin{proof}
Let $k = j - i$. The $\eps$-synchronization string property of $S$ guarantees that 
\[ED\left(S[i-k, i),S[i,j)\right) > (1-\eps) 2k.\]
Note that this holds even if $i-k<1$. To finish the proof we note that the maximum in the definition of RSD includes the term $\frac{ED\left(S[i-k, i),S[i,j)\right)}{2k} > 1-\eps$, which implies that $RSD(S[1,i],S[1,j]) > 1-\eps$.
\end{proof}

}

We now recall synchronization string based techniques and relevant lemmas from \cite{haeupler2017synchronization, haeupler2017synchronization2:ARXIV} which we will be of use here. In short, synchronization strings allow communicating parties to protect against synchronization errors by indexing their messages without blowing up the communication rate. The general idea of coding schemes introduced and utilized in \cite{haeupler2017synchronization, haeupler2017synchronization2:ARXIV}, is to index any communicated symbol in the sender side and then \emph{guess} the actual position of received symbols on the other end using the attached indices.

A straightforward candidate for such technique is to attach $1, \cdots, n$ to communicated symbols where $n$ indicates the rounds of communication. However, this trivial indexing scheme would not lead to an efficient solution as it requires assigning a $\log n$-sized space to indexing symbols. This shortcoming accentuates a natural trade-off between the size of the alphabet among which indexing symbols are chosen and the accuracy of the guessing procedure on the receiver side.

Haeupler and Shahrasbi~\cite{haeupler2017synchronization} introduce and discuss $\eps$-synchronization strings as well-fitting candidates for this matter. This family of strings, parametrized by $\eps$, are over alphabets of constant size in terms of communication length $n$ and dependent merely on parameter $\eps$. $\eps$-synchronization strings can convert any adversarial $k$ synchronization errors into hamming-type errors. The extent of disparity between the number translated hamming-type errors and $k$ can be controlled by parameter $\eps$. 

Imagine Alice and Bob as two parties communicating over a channel suffering from up to $\delta$-fraction of adversarial insertions and deletions. Suppose Alice sends a string $S$ of length $n$ to Bob. On the other end of the communication, Bob will receive a distorted version of $S$ as adversary might have inserted or deleted a number of symbols. 
A symbol which is sent by Alice and is received by Bob without being deleted by the adversary is called a \emph{successfully transmitted} symbol.

Assume that Alice and Bob both know string $S$ a priori. Bob runs an algorithm to determine the actual index of each of the symbols he receives, in other words, to guess which element of $S$ they correspond to. Such algorithm has to return an number in $[1, n]$ or ``I don't know'' for any symbol of $S_\tau$. We call such an algorithm an \emph{$(n, \delta)$-indexing algorithm}.

Ideally, a indexing algorithm is supposed to correctly figure out the indices of as many successfully transmitted symbols as possible. The measure of \emph{misdecodings} has been introduced in~\cite{haeupler2017synchronization} to evaluate the quality of a $(n, \delta)$-indexing algorithm as the number of successfully transmitted symbols that an algorithm might not decoded correctly. 
An indexing algorithm is called to be \emph{streaming} if its output for a particular received symbol depends only on the symbols that have been received before it.

Haeupler and Shahrasbi~\cite{haeupler2017synchronization} introduce and discuss $\eps$-synchronization strings along with several decoding techniques for them.

\begin{definition}[$\eps$-Synchronization String]\label{def:synCode}
String $S \in \Sigma^n$ is an $\eps$-synchronization string if for every $1 \leq i < j < k \leq n + 1$ we have that $ED\left(S[i, j),S[j, k)\right) > (1-\eps) (k-i)$. We call the set of prefixes of such a string an $\eps$-synchronization code. 
\end{definition}

We will make use of the global decoding algorithm from~\cite{haeupler2017synchronization} described as follows.
\begin{theorem}[Theorems  and 6.14 from \cite{haeupler2017synchronization}]\label{thm:globalDecoding}
There is a decoding algorithm for an $\eps$-synchronization string of length $n$ which guarantees decoding with up to $O(n\sqrt{\eps})$ misdecodings and runs in $O(n^2/\sqrt{\eps})$ time.
\end{theorem}

\begin{theorem}[Theorem 4.1 from~\cite{haeupler2017synchronization}]\label{thm:mainECC}
Given a synchronization string $S$ over alphabet $\Sigma_S$ with an (efficient) decoding algorithm $\mathcal{D}_S$ guaranteeing at most $k$ misdecodings and decoding complexity $T_{\mathcal{D}_{S}}(n)$ and an (efficient) ECC $\mathcal{C}$ over alphabet $\Sigma_{\mathcal{C}}$ with rate $R_{\mathcal{C}}$, encoding complexity $T_{\mathcal{E}_{\mathcal{C}}}$, and decoding complexity $T_{\mathcal{D}_{\mathcal{C}}}$ that corrects up to $n\delta + 2k$ half-errors, one obtains an insdel code that can be (efficiently) decoded from up to $n\delta$ insertions and deletions. The rate of this code is at least
$$\frac{R_{\mathcal{C}}}{1 + \frac{\log \Sigma_S}{\log \Sigma_{\mathcal{C}}}}$$
The encoding complexity remains $T_{\mathcal{E}_{\mathcal{C}}}$, the decoding complexity is $T_{\mathcal{D}_{\mathcal{C}}} + T_{\mathcal{D}_{S}}(n)$ and the preprocessing complexity of constructing the code is the complexity of constructing $\mathcal{C}$ and $S$.
\end{theorem}

\global\def\ExtendedPrelimSync{\color{red}
To review Haeupler et. al.~\cite{haeupler2017synchronization, haeupler2017synchronization2:ARXIV} techniques more accurately, we start reviewing their basic definitions and formal problem statements. We being with the notion of string matching, originally introduced by Braverman et. al. \cite{braverman2017coding}.

\begin{definition}[String matching]
Suppose that $c$ \allowbreak and $c'$ are two strings in $\Sigma^*$, and suppose that $*$ is a symbol not in $\Sigma$. Next, suppose that there exist two strings $\tau_1$ and $\tau_2$ in $\left(\Sigma \cup \{*\}\right)^*$ such that $|\tau_1| = |\tau_2|$, $del\left(\tau_1\right) = c$, $del(\tau_2) = c'$, and $\tau_1[i] \approx \tau_2[i]$ for all $i \in \left\{1, \dots, |\tau_1|\right\}$. Here, $del$ is a function that deletes every $*$ in the input string and $a \approx b$ if $a = b$ or one of $a$ or $b$ is $*$. Then we say that $\tau = \left(\tau_1, \tau_2\right)$ is a \emph{string matching} between $c$ and $c'$ (denoted $\tau: c \to c'$). We furthermore denote with $sc\left(\tau_i\right)$ the number of $*$'s in $\tau_i$.
\end{definition}

Imagine Alice and Bob as two parties communicating over a channel suffering from up to $\delta$-fraction of adversarial insertions and deletions. Suppose Alice sends a string $S$ of length $n$ to Bob. On the other end of the communication, Bob will receive a distorted version of $S$ as adversary might have inserted or deleted a number of symbols. Note that adversary's actions can be characterized by a string matching $\tau$ from $S$ to the distorted version of $S$ received by Bob that we will denote by $S_\tau$.

A symbol which is sent by Alice and is received by Bob without being deleted by the adversary is called a \emph{successfully transmitted} symbol. More formally, $S_\tau[j]$ for $\tau=(\tau_1, \tau_2)$ is a successfully transmitted symbol if there exist a $k$ such that $|del(\tau_2[1,k])=j$ and $\tau_1[k] = \tau_2[k]$.

Assume that Alice and Bob both know string $S$ a priori. Bob runs an algorithm to determine the actual index of each of the symbols he receives, in other words, to guess which element of $S$ they correspond to. Such algorithm has to return an number in $[1, n]$ or ``I don't know'' for any symbol of $S_\tau$. We call such an algorithm an \emph{$(n, \delta)$-indexing algorithm}. A formal definition of such algorithm is defined as follows.

\begin{definition}[$(n,\delta)$-Indexing Algorithm] The pair $(S, \mathcal{D}_S)$ consisting of a string $S \in \Sigma^n$ and an algorithm $\mathcal{D}_S$ is called an $(n,\delta)$-indexing algorithm over alphabet $\Sigma$ if for any set of $n\delta$ insertions and deletions altering the string $S$ to a string $S_{\tau}$, the algorithm $\mathcal{D}_S(S_{\tau})$ outputs either $\top$ or an index between 1 and $n$ for every symbol in $S_{\tau}$.
\end{definition}

Ideally, a indexing algorithm is supposed to correctly figure out the indices of as many successfully transmitted symbols as possible. Therefore, Haeupler and Shahrasbi~\cite{haeupler2017synchronization} introduced the measure of \emph{misdecodings} to evaluate the quality of a $(n, \delta)$-indexing algorithm as the number of successfully transmitted symbols which are not decoded correctly. This includes outputting an incorrect guessed index or $\top$ for a successfully transmitted symbol. Formally, we say an $(n, \delta)$-indexing algorithm is guaranteed to have up to $k$ misdecodings if for any possible $n\delta$ or less adversarial insertions or deletions, the outcome of the decoding algorithm on the receiver side does not contain more than $k$ misdecodings.
An indexing algorithm is called to be \emph{streaming} if its output for a particular received symbol $S_\tau[j]$ depends only on symbols that have been received sooner than it, i.e., $S_\tau[1, j]$.

Haeupler and Shahrasbi~\cite{haeupler2017synchronization} introduce and discuss $\eps$-synchronization strings along with several decoding techniques that together form well-fitting candidates for indexing algorithm.

\begin{definition}[$\eps$-Synchronization String]\label{def:synCode}
String $S \in \Sigma^n$ is an $\eps$-synchronization string if for every $1 \leq i < j < k \leq n + 1$ we have that $ED\left(S[i, j),S[j, k)\right) > (1-\eps) (k-i)$. We call the set of prefixes of such a string an $\eps$-synchronization code. 
\end{definition}

Haeupler and Shahrasbi provide an efficient randomized construction of $\eps$-synchronization strings of arbitrary length over a constant-sized alphabet.

\begin{lemma}[From \cite{haeupler2017synchronization}]\label{lemma:FiniteSyncConstruction}
There exists a randomized algorithm which, for any $\eps >0$, constructs a $\eps$-synchronization string of length $n$ over an alphabet of size $O(\eps^{-4})$ in expected time $O(n^5)$.
\end{lemma}

They prove the following useful property which leads them to an indexing algorithm.

\begin{lemma}[From \cite{haeupler2017synchronization}]
Let $S \in \Sigma^n$ be an $\eps$-synchronization string and let $S_{\tau}[1,j]$ be a prefix of $S_{\tau}$. Then there exists at most one index $i \in [n]$ such that the suffix distance between $S_{\tau}[1,j]$ and $S[1,i]$, denoted by $SD(S_{\tau}[1,j],S[1,i])$ is at most $1-\eps$.
\end{lemma}

This lemma suggests a simple $(n,\delta)$-indexing algorithm given an input prefix $S_{\tau}[1,j]$: Use dynamic programming to search over all prefixes of $S$ for the one with the smallest suffix distance from $S_{\tau}[1,j]$. Haeupler and Shahrasbi present a dynamic programming algorithm which is efficient and results in a small number of misdecodings, and described in the following theorem.

\begin{theorem}[From \cite{haeupler2017synchronization}]\label{thm:RSPDmisdecodings}
Let $S \in \Sigma^n$ be an $\eps$-synchronization string that is sent over an insertion-deletion channel with a $\delta$ fraction of insertions and deletions. There exists a streaming $(n,\delta)$-indexing algorithm that returns a solution with $\frac{c_i}{1-\eps} + \frac{c_d\eps}{1-\eps}$ misdecodings. The algorithm runs in time $O(n^5)$, spending $O(n^4)$ on each received symbol.
\end{theorem}

We now quickly review main theorems of Haeupler et. al.~\cite{haeupler2017synchronization, haeupler2017synchronization2:ARXIV} on how indexing algorithms can be of use to design insertion-deletion codes and even more strongly, simulate hamming-type channels over given insertion-deletion channels.

\begin{theorem}[Theorem 4.1 from~\cite{haeupler2017synchronization}]\label{thm:mainECC}
Given a synchronization string $S$ over alphabet $\Sigma_S$ with an (efficient) decoding algorithm $\mathcal{D}_S$ guaranteeing at most $k$ misdecodings and decoding complexity $T_{\mathcal{D}_{S}}(n)$ and an (efficient) ECC $\mathcal{C}$ over alphabet $\Sigma_{\mathcal{C}}$ with rate $R_{\mathcal{C}}$, encoding complexity $T_{\mathcal{E}_{\mathcal{C}}}$, and decoding complexity $T_{\mathcal{D}_{\mathcal{C}}}$ that corrects up to $n\delta + 2k$ half-errors, one obtains an insdel code that can be (efficiently) decoded from up to $n\delta$ insertions and deletions. The rate of this code is at least
$$\frac{R_{\mathcal{C}}}{1 + \frac{\log \Sigma_S}{\log \Sigma_{\mathcal{C}}}}$$
The encoding complexity remains $T_{\mathcal{E}_{\mathcal{C}}}$, the decoding complexity is $T_{\mathcal{D}_{\mathcal{C}}} + T_{\mathcal{D}_{S}}(n)$ and the preprocessing complexity of constructing the code is the complexity of constructing $\mathcal{C}$ and $S$.
\end{theorem}

Further, Haeupler et. al.~\cite{haeupler2017synchronization2:ARXIV} showed that indexing algorithms that satisfy streaming property can simulate ordinary corruption channels over given insertion deletion channels. In particular, they have shown that any streaming decoding algorithm that forms a solution to $(n, \delta)$-indexing algorithm with a $\eps$-synchronization string guaranteeing $O\left(\frac{n\delta}{1-\eps}\right)$ many misdecodings can be used to get following channel simulation results. We refer readers to~\cite{haeupler2017synchronization2:ARXIV} for the details of such simulations.

\begin{theorem}\label{thm:simulationRecap}(Based on Theorems 3.3, 3.5, 3.11, and 3.13 of~\cite{haeupler2017synchronization2:ARXIV})
\begin{enumerate}
\item[(a)] Suppose that $n$ rounds of a one-way/interactive insertion-deletion channel over an alphabet $\Sigma$ with a $\delta$ fraction of insertions and deletions are given. Using an $\eps$-synchronization string over an alphabet $\Sigma_{syn}$, it is possible to simulate $n\left(1-O_\eps(\delta)\right)$ rounds of a one-way/interactive corruption channel over $\Sigma_{sim}$ with at most $O_\eps\left(n\delta\right)$ symbols corrupted so long as $|\Sigma_{sim}| \times |\Sigma_{syn}| \le |\Sigma|$. 
\item[(b)] Suppose that $n$ rounds of a binary one-way/interactive insertion-deletion channel with a $\delta$ fraction of insertions and deletions are given. It is possible to simulate 
$n(1-\Theta( \sqrt{\delta\log(1/\delta)}))$
 rounds of a binary one-way/interactive corruption channel 
 with $\Theta(\sqrt{\delta\log(1/\delta)})$ fraction of corruption errors between two parties over the given channel.
\end{enumerate}
Having the synchronization string utilized in the simulation as a pre-processed component, upon sending each symbol, the simulations introduced above take $O(1)$ for sending/starting party of one-way/interactive communications respectively. Further, the simulation spends polynomial time in terms of $n$ upon arrival of each symbol on the other side.
\end{theorem}
}
}
\fullOnly{\PrelimsSynchStrings}
\section{Highly Explicit Constructions of Long-Distance and Infinite $\eps$-Synchronization Strings}\label{sec:sync_construction}
We start this section by introducing a generalized notion of synchronization strings in Section~\ref{sec:LongDistSync} and then provide a deterministic efficient construction for them in Section~\ref{sec:LLLConstruction}. In Section~\ref{sec:BoostingStepII}, we provide a boosting step which speeds up the construction to linear time in Theorem~\ref{thm:linearTimeConstruction}. In Section~\ref{sec:BoostingStepI}, we use the linear time construction to obtain a linear-time high-distance insdel code (Theorem~\ref{thm:insdelCode}) and then use another boosting step to obtain a highly-explicit linear-time construction for long-distance synchronization strings in Theorem~\ref{thm:FiniteLongDistConstructionNEW}. We provide similar construction for infinite synchronization strings in Section~\ref{sec:infiniteConstruction}. A pictorial representation of the flow of theorems and lemmas in this section can be found in Figure~\ref{fig:flow}.
\begin{figure}
\centering
\includegraphics[scale=.38]{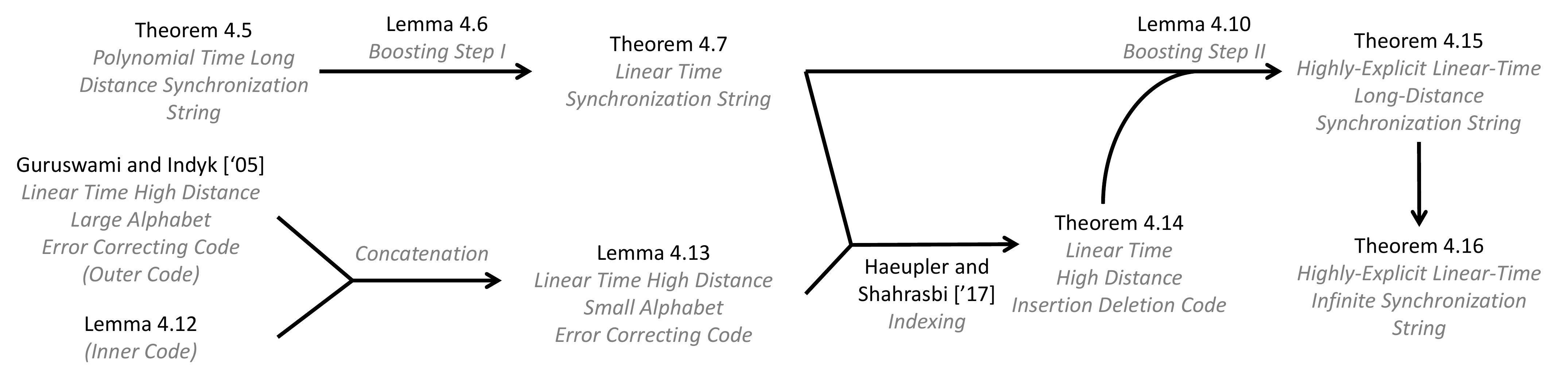}
\caption{Schematic flow of Theorems and Lemmas of Section~\ref{sec:sync_construction}}
\label{fig:flow}
\end{figure}

\subsection{Long-Distance Synchronization Strings}\label{sec:LongDistSync}
The existence of synchronization strings is proven in~\cite{haeupler2017synchronization} using an argument based on Lov\'{a}sz local lemma. This lead to an efficient randomized construction for synchronization strings which cannot be easily derandomized. Instead, the authors introduced the weaker notion of self-matching strings and gave a deterministic construction for them. Interestingly, in this paper we introduce a revised notion, denoted by \emph{$f(l)$-distance $\eps$-synchronization strings}, which generalizes $\eps$-synchronization strings and allows for a deterministic construction.

Note that the synchronization string property poses a requirement on the edit distance of neighboring substrings. $f(l)$-distance $\eps$-synchronization string property extends this requirement to any pair of intervals that are nearby. More formally, any two intervals of aggregated length $l$ that are of distance $f(l)$ or less have to satisfy the edit distance property in this generalized notion.

\begin{definition}[$f(l)$-distance $\eps$-synchronization string]\label{def:distSynchStr}
String $S \in \Sigma^n$ is an $f(l)$-distance $\eps$-synchronization string if for every $1 \leq i<j\leq i'<j' \leq n + 1$ we have that $ED\left(S[i,j),S[i',j')\right) > (1-\eps) (k-i)$ if $i'-j \leq f(l)$ where $l=j+j'-i-i'$.
\end{definition}

It is noteworthy to mention that the constant function $f(l)=0$ gives the original $\eps$-synchronization strings. Haeupler and Shahrasbi~\cite{haeupler2017synchronization} have studied the existence and construction of synchronization strings for this case. In particular, they have shown that arbitrarily long $\eps$-synchronization strings exist over an alphabet that is polynomially large in terms of $\eps^{-1}$. 
Besides $f(l)=0$, there are other several other functions that might be of interest in this context.

One can show that, as we do in Appendix~\ref{app:alphabetSizeVSDistance}, that for any polynomial function $f(l)$, arbitrarily long $f(l)$-distance $\eps$-synchronization strings exist over alphabet sizes that are polynomially large in terms of $\eps^{-1}$. Also, for exponential functions, these strings exist over exponentially large alphabets in terms of $\eps^{-1}$ but not over sub-exponential alphabet sizes. Finally, if function $f$ is super-exponential, $f(l)$-distance $\eps$-synchronization strings do not exist over any alphabet whose size is independent of $n$.

While studying existence, construction, and alphabet sizes of $f(l)$-distance $\eps$-synchronization strings might be of interest by its own, we will show that having synchronization string edit distance guarantee for pairs of intervals that are exponentially far in terms of their aggregated length is of significant interest as it leads to improvements over applications of ordinary synchronization strings described in~\cite{haeupler2017synchronization, haeupler2017synchronization2:ARXIV} from several aspects. Even though distance function $f(l)=c^l$ provides such property, throughout the rest of this paper, we will focus on a variant of it, i.e., $f(l) = n \cdot \mathbbm{1}_{l>c\log n}$ which allows polynomial-sized alphabet.
$\mathbbm{1}_{l>c\log n}$ is the indicator function for $l>c\log n$, i.e., one if $l>c\log n$ and zero otherwise

To compare distance functions $f(l) = c^l$ and $f(l) = n \cdot \mathbbm{1}_{l>c\log n}$, note that the first one allows intervals to be exponentially far away in their total length. In particular, intervals of length $l>c\log n$ or larger can be arbitrarily far away. The second function only asks for the guarantee over large intervals and does not strengthen the $\eps$-synchronization property for smaller intervals. We refer to the later as \emph{$c$-long-distance $\eps$-synchronization string} property.

\begin{definition}[$c$-long-distance $\eps$-synchronization strings]
We call $n \cdot \mathbbm{1}_{l>c\log n}$-distance $\eps$-synchronization strings \emph{$c$-long-distance $\eps$-synchronization strings}.
\end{definition}

\subsection{Polynomial Time Construction of Long-Distance Synchronization Strings}\label{sec:LLLConstruction}

An LLL-based proof for existence of ordinary synchronization strings has benn provided by~\cite{haeupler2017synchronization}.
Here we provide a similar technique along with the deterministic algorithm for Lov\'{a}sz local lemma from Chandrasekaran et al.\cite{chandrasekaran2013deterministic} to prove the existence and give a deterministic polynomial-time construction of strings that satisfy this quality over an alphabet of size $\eps^{-O(1)}$. 

Before giving this proof right away, we first show a property of the these strings which allows us to simplify the proof and, more importantly, get a deterministic algorithm using deterministic algorithms for Lov\'{a}sz local lemma from Chandrasekaran et al.\cite{chandrasekaran2013deterministic}. 

\begin{lemma}\label{lem:longdistancereduction}
If $S$ is a string and there are two intervals $i_1<j_1\leq i_2 < j_2$ of total length $l = j_1 - i_1 + j_2 - i_2$ and $ED(S[i_1,j_1), S[i_2,j_2)) \le (1-\eps) l$ then there also exists intervals $i_1 \leq i'_1<j'_1\leq i'_2 < j'_2 \leq i_2$ of total length $l' \in \{\ceil{l/2}-1,\ceil{l/2}, \ceil{l/2}+1\}$ with $ED(S[i'_1,j'_1),S[i'_2,j'_2)) \leq (1-\eps) l'$.
\end{lemma}
\global\def\ProofOfLemLongdistancereduction{
\shortOnly{\subsection{Proof of Lemma~\ref{lem:longdistancereduction}}}
\begin{proof}
As $ED(S[i_1,j_1),[i_2,j_2)) \le (1-\eps) l$, there has to be a monotone matching 
$M = \{(a_1,b_1), \cdots, (a_m,b_m)\}$
 from $S[i_1,j_1)$ to $S[i_2,j_2)$ of size $m \ge \frac{\eps l}{2}$. Let $1\le i\le m$ be the largest number such that $|S[i_1,a_i]| + |S[i_2,b_i]| \le \ceil{l/2}$. It is easy to verify that there are integers $a_i < k_1 \le a_{i+1}$ and $b_i < k_2 \le b_{i+1}$ such that $|S[i_1, k_1)| + |S[i_2, k_2)| \in \{\ceil{l/2} - 1, \ceil{l/2}\}$.
 
Therefore, we can split the pair of intervals $\left(S[i_1,j_1), S[i_2,j_2)\right)$ into two pairs of intervals $\left(S[i_1,k_1), S[i_2,k_2)\right)$ and $\left(S[k_1,j_1), S[k_2,j_2)\right)$ such that each pair of the matching $M$ falls into exactly one of these pairs. Hence, in at least one of those pairs, the size of the matching is larger than $\frac{\eps}{2}$ times the total length.
This gives that the edit distance of those pairs is less than $1-\eps$ and finishes the proof.
\end{proof}
}\fullOnly{\ProofOfLemLongdistancereduction}

Lemma~\ref{lem:longdistancereduction} shows that if there is a pair of intervals of total length $l$ that have small relative edit distance, we can find a pair of intervals of size 
$\{\ceil{l/2}-1,\ceil{l/2}, \ceil{l/2}+1\}$ which have small relative edit distance as well. Now, let us consider a string $S$ with a pair of intervals that violate the $c$-long distance $\eps$-synchronization property. If the total length of the intervals exceed $2c \log  n$, using Lemma~\ref{lem:longdistancereduction} we can find another pair of intervals of almost half the total length which still violate the $c$-long distance $\eps$-synchronization property. Note that as their total length is longer than $c\log n$, we do not worry about the distance of those intervals. Repeating this procedure, we can eventually find a pair of intervals of a total length between $c\log n$ and $2c\log n$ that violate the $c$-long distance $\eps$-synchronization property. More formally, we can derive the following statement by Lemma~\ref{lem:longdistancereduction}.

\begin{corollary}\label{cor:shortDistanceSufficient}
If $S$ is a string which satisfies the $c$-long-distance $\eps$-synchronization property for any two non-adjacent intervals of total length $2c\log n$ or less, then it satisfies the property for all pairs of non-adjacent intervals.
\end{corollary}
\global\def\ProofOfCorShortDistanceSufficient{
\shortOnly{\subsection{Proof of Corollary~\ref{cor:shortDistanceSufficient}}}
\begin{proof}
Suppose, for the sake of contradiction, that  there exist two intervals of total length $2\log_c n$ or more that violate the $c$-long-distance $\eps$-synchronization property. Let $[i_1, j_1)$ and $[i_2, j_2)$ where $i_1<j_1\leq i_2 < j_2$ be two intervals of the smallest total length $l = j_1 - i_1 + j_2 - i_2$ larger than $2\log_c n$ (breaking ties arbitrarely) for which $ED(S[i_1,j_1),[i_2,j_2)) \leq (1 - \eps) l$. By Lemma~\ref{lem:longdistancereduction} there exists two intervals $[i'_1, j'_1)$ and $[i'_2, j'_2)$ where $i'_1<j'_1\leq i'_2 < j'_2$ of total length $l' \in [l/2,l)$ with $ED(S[i'_1,j'_1),[i'_2,j'_2)) \leq (1 - \eps) l$. If $l' \le 2\log_c n$, the assumption of $c$-long-distance $\eps$-synchronization property holding for intervals of length $2 \log_c n$ or less is contradicted. Unless, $l' > 2\log_c n$ that contradicts the minimality of our choice of $l$.
\end{proof}
}\fullOnly{\ProofOfCorShortDistanceSufficient}

\begin{theorem}\label{thm:polyConstruction}
For any $0<\eps<1$ and every $n$ there is a deterministic $n^{O(1)}$ time algorithm for computing a $c=O(1/\eps)$-long-distance $\eps$-synchronization string over an alphabet of size $O(\eps^{-4})$.
\end{theorem}
\shortOnly{
\begin{proof}[Proof Sketch]
We will make use of the Lov\'{a}sz local lemma and deterministic algorithms proposed for it in~\cite{chandrasekaran2013deterministic}. We generate a random string $R$ over an alphabet of size $|\Sigma|=O(\eps^{-2})$ and define bad event $B_{i_1, l_1, i_2, l_2}$ as the event of intervals $[i_1, i_1+l_1)$ and $[i_2, i_2+l_2)$ violating the $O(1/\eps)$-long-distance synchronization string property over intervals of total length $2/\eps^2$ or more. 
We will show that for large enough $n$, with non-zero probability, none of these bad events happen and an instance can be found in polynomial time. To turn this string into a $c=O(1/\eps)$-long-distance $\eps$-synchronization strings, we simply concatenate it with a string consisting of repetitions of $1, \cdots, 2\eps^{-2}$ which takes care of the edit distance requirement for neighboring intervals with total length smaller than $2\eps^{-2}$.
\end{proof}
}
\global\def\ProofOfLLLConstruction{
\fullOnly{\begin{proof}}
\shortOnly{\subsection{Proof of Theorem~\ref{thm:polyConstruction}}\begin{proof}}
To proof this, we will make use of the Lov\'{a}sz local lemma and deterministic algorithms proposed for it in~\cite{chandrasekaran2013deterministic}. We generate a random string $R$ over an alphabet of size $|\Sigma|=O(\eps^{-2})$ and define bad event $B_{i_1, l_1, i_2, l_2}$ as the event of intervals $[i_1, i_1+l_1)$ and $[i_2, i_2+l_2)$ violating the $O(1/\eps)$-long-distance synchronization string property over intervals of total length $2/\eps^2$ or more. In other words, $B_{i_1, l_1, i_2, l_2}$ occurs if and only if $ED(R[i_1, i_1+l_1), R[i_2, i_2 + l_2)) \le (1-\eps)(l_1 + l_2)$. 
Note that by the definition of $c$-long-distance $\eps$-synchronization strings, 
$B_{i_1, l_1, i_2, l_2}$
is defined for $(i_1, l_1, i_2, l_2)$s where either $l_1+l_2\ge c\log n$ and $i_1+l_1\le i_2$ or $1/\eps^2<l_1+l_2 < c\log n$ and $i_2 = i_1+l_1$. 
We aim to show that for large enough $n$, with non-zero probability, none of these bad events happen. This will prove the existence of a string that satisfies $c=O(1/\eps)$-long-distance $\eps$-synchronization strings for all pairs of intervals that are of total length $2/\eps^2$ or more. To turn this string into a $c=O(1/\eps)$-long-distance $\eps$-synchronization strings, we simply concatenate it with a string consisting of repetitions of $1, \cdots, 2\eps^{-2}$, i.e., $1,2,\cdots,2\eps^{-2}, 1,2,\cdots,2\eps^{-2},\cdots$. This string will take care of the edit distance requirement for neighboring intervals with total length smaller than $2\eps^{-2}$.

Note that using Lemma~\ref{lem:longdistancereduction}, by a similar argument as in Claim~\ref{cor:shortDistanceSufficient}, we only need to consider bad events where 
$l_1 + l_2 \le 2c\log n$. As the first step, note that $B_{i_1, l_1, i_2, l_2}$ happens only if there is a common subsequence of length $\eps (l_1+l_2) /2$ or more between $R[i_1, i_1 + l_1)$ and $R[i_2, i_2 + l_2)$. Hence, the union bound gives that

\begin{eqnarray*}
\Pr\left\{B_{i_1, l_1, i_2, l_2}\right\} &\le& {l_1 \choose \eps(l_1+l_2)/2}{l_1 \choose \eps(l_1+l_2)/2}{\left|\Sigma\right|}^{-\frac{\eps (l_1+l_2)}{2}}\\
&\le& \left(\frac{l_1 e}{\eps(l_1+l_2)/2}\right)^{\eps(l_1+l_2)/2} \left(\frac{l_2 e}{\eps(l_1+l_2)/2}\right)^{\eps(l_1+l_2)/2}{\left|\Sigma\right|}^{-\frac{\eps (l_1+l_2)}{2}}\\
&=&\left(\frac{2e\sqrt{l_1l_2}}{\eps(l_1+l_2)\sqrt{|\Sigma|}}\right)^{\eps(l_1+l_2)}\\
&\le&  \left(\frac{el}{\eps l \sqrt{|\Sigma|}}\right)^{\eps l} = \left(\frac{e}{\eps  \sqrt{|\Sigma|}}\right)^{\eps l}
\end{eqnarray*}
where $l=l_1 + l_2$.
In order to apply LLL, we need to find real numbers $x_{i_1, l_1, i_2, l_2} \in [0, 1]$ such that for any $B_{i_1, l_1, i_2, l_2}$
\begin{equation}\label{eqn:LLLcondition}
\Pr\{B_{i_1, l_1, i_2, l_2}\} \le x_{i_1, l_1, i_2, l_2} \prod_
{\left[S[i_1, i_1+l_1)\cup S[i_2, i_2+l_2)\right]\cap[S[i'_1, i'_1+l'_1)\cup S[i'_2, i'_2+l'_2)]\neq \emptyset}
 (1-x_{i'_1, l'_1, i'_2, l'_2})
\end{equation}
 
We eventually want to show that our LLL argument satisfies the conditions required for polynomial-time deterministic algorithmic LLL specified in~\cite{chandrasekaran2013deterministic}. 
Namely, it suffices to certify two other properties in addition to~\eqref{eqn:LLLcondition}. The first additional requirement is to have each bad event in LLL depend on up to logarithmically many variables and the second is to have~\eqref{eqn:LLLcondition} hold with a constant exponential slack. The former is clearly true as our bad events consist of pairs of intervals each of which is of a length between $c\log n$ and $2c\log n$. To have the second requirement, instead of~\eqref{eqn:LLLcondition} we find $x_{i_1, l_1, i_2, l_2} \in [0, 1]$ that satisfy the following stronger property.
\begin{equation}\label{eqn:LLLcondition2}
\Pr\{B_{i_1, l_1, i_2, l_2}\} \le \left[x_{i_1, l_1, i_2, l_2} \prod_
{\left[S[i_1, i_1+l_1)\cup S[i_2, i_2+l_2)\right]\cap[S[i'_1, i'_1+l'_1)\cup S[i'_2, i'_2+l'_2)]\neq \emptyset}
 (1-x_{i'_1, l'_1, i'_2, l'_2})\right]^{1.01}
\end{equation}
Any small constant can be used as slack. We pick 1.01 for the sake of simplicity. We propose $x_{i_1, l_1, i_2, l_2} = D^{-\eps (l_1 + l_2)}$ for some $D> 1$ to be determined later. $D$ has to be chosen such that for any $i_1, l_1, i_2, l_2$ and $l=l_1 + l_2$:
\begin{eqnarray}
\left(\frac{e}{\eps  \sqrt{|\Sigma|}}\right)^{\eps l} &\le& \left[D^{-\eps l} \prod_
{[S[i_1, i_1+l_1)\cup S[i_2, i_2+l_2)]\cap[S[i'_1, i'_1+l'_1)\cup S[i'_2, i'_2+l'_2)]\neq \emptyset}
 \left(1-D^{-\eps (l'_1 + l'_2)}\right)\right]^{1.01}\label{eq:LLLMiddleStep}
\end{eqnarray}

Note that:
\begin{eqnarray}
&&D^{-\eps l} \prod_{[S[i_1, i_1+l_1)\cup S[i_2, i_2+l_2)]\cap[S[i'_1, i'_1+l'_1)\cup S[i'_2, i'_2+l'_2)]\neq \emptyset}
 \left(1-D^{-\eps (l'_1 + l'_2)}\right)\allowdisplaybreaks \\
&\ge& D^{-\eps l} \prod_{l'=c\log n}^{2c\log n}\prod_{l'_1=1}^{l'}
\left(1-D^{-\eps l'}\right)^{\left[(l_1+l'_1)+(l_1+l'_2)+(l_2+l'_1)+(l_2+l'_2)\right] n}
\nonumber\\
&&\times \prod_{l''=1/\eps^2}^{c\log n} \left(1-D^{-\eps l''}\right)^{l+l''}\label{eqn:countPairOfIntervals}\allowdisplaybreaks \\
 &=& D^{-\eps l} \prod_{l'=c\log n}^{2c\log n}\prod_{l'_1=1}^{l'}
 \left(1-D^{-\eps l'}\right)^{4(l+l') n}
 \times \prod_{l''=1/\eps^2}^{c\log n} \left(1-D^{-\eps l''}\right)^{l+l''}
 \allowdisplaybreaks \\&=& 
 D^{-\eps l} \prod_{l'=c\log n}^{2c\log n}
 \left(1-D^{-\eps l'}\right)^{4l'(l+l') n}
  \times \left[\prod_{l''=1/\eps^2}^{c\log n} \left(1-D^{-\eps l''}\right)\right]^{l}
  \times \prod_{l''=1/\eps^2}^{c\log n} \left(1-D^{-\eps l''}\right)^{l''}\allowdisplaybreaks \\
 &\ge&D^{-\eps l} \left(1-\sum_{l'=c\log n}^{2c\log n}\left(4l'(l+l') n\right)D^{-\eps l'}\right)
\nonumber\\&&
\times \left[1-\sum_{l''=1/\eps^2}^{c\log n} D^{-\eps l''}\right]^{l}
 \times  \left(1-\sum_{l''=1/\eps^2}^{c\log n}l''D^{-\eps l''}\right)
\label{eqn:linearization}\allowdisplaybreaks \\
  &\ge&D^{-\eps l} \left(1-\sum_{l'=c\log n}^{2c\log n}\left(4\cdot 2c\log n(2c\log n+2c\log n) n\right)D^{-\eps l'}\right)\\
&& 
\times \left[1-\sum_{l''=1/\eps^2}^{\infty} D^{-\eps l''}\right]^{l}
 \times  \left(1-\sum_{l''=1/\eps^2}^{\infty}l''D^{-\eps l''}\right)
\allowdisplaybreaks \\
  &=&D^{-\eps l} \left(1-\sum_{l'=c\log n}^{2c\log n}\left(32c^2n\log^2 n\right)D^{-\eps l'}\right) 
\times \left[1-\frac{D^{-\eps\cdot1/\eps^2}}{1-D^{-\eps }}\right]^{l}
\nonumber\\&&
 \times  \left(1-\frac{D^{-\eps\cdot 1/\eps^2}(D^{-\eps}+1/\eps^2-D^{-\eps}/\eps^2)}{(1-D^{-\eps})^2}\right)
 \allowdisplaybreaks \\&\ge& D^{-\eps l} \left(1-32c^3n\log^3 nD^{-\eps c\log n}\right)
 \left[1-\frac{D^{-1/\eps}}{1-D^{-\eps }}\right]^{l}
\nonumber\\&&
\times \left(1-\frac{D^{-1/\eps}(D^{-\eps}+1/\eps^2-D^{-\eps}/\eps^2)}{(1-D^{-\eps})^2}\right)\label{eqn:BadEventProbLowerBound}
\end{eqnarray}
To justify equation~\eqref{eqn:countPairOfIntervals}, note that there are two kinds of bad events that might intersect $B_{i_1, l_1, i_2, l_2}$. The first product term is considering all pairs of long intervals of length $l'_1$ and $l'_2$ where $l_1 + l_2 \ge c\log n$ that overlap a fixed pair of intervals of length $l_1$ and $l_2$. The number of such intervals is at most $\left[(l_1+l'_1)+(l_1+l'_2)+(l_2+l'_1)+(l_2+l'_2)\right] n$. The second one is considering short neighboring pairs of intervals ($\eps^{-2}\le l''=l''_1+l''_2 \le c\log n$).

Equation~\eqref{eqn:linearization} is a result of the following inequality for $0<x, y<1$:
$$(1-x)(1-y) > 1-x-y.$$

We choose $D=2$ and $c=2/\eps$. Note that 
$\lim_{\eps\rightarrow 0}\frac{2^{-1/\eps}(2^{-\eps}+1/\eps^2-2^{-\eps}/\eps^2)}{(1-2^{-\eps})^2}=0$. So, for small enough $\eps$, $\frac{2^{-1/\eps}}{1-2^{-\eps}} < \frac{1}{2}$.
Also, for $D=2$ and $c=2/\eps$,
$$32c^3n\log^3 nD^{-\eps c\log n} = \frac{2^8}{\eps^3}\cdot \frac{\log ^3 n}{n} = o(1).$$
Finally, one can verify that for small enough $\eps$, 
$1-\frac{2^{-1/\eps}}{1-2^{-\eps}}
 > 2^{-\eps}$.
Therefore, for sufficiently small $\eps$ and sufficiently large $n$, \eqref{eqn:BadEventProbLowerBound} is satisfied if the following is satisfied.
\begin{eqnarray}
&&D^{-\eps l} \prod_{[S[i_1, i_1+l_1)\cup S[i_2, i_2+l_2)]\cap[S[i'_1, i'_1+l'_1)\cup S[i'_2, i'_2+l'_2)]\neq \emptyset}
 \left(1-D^{-\eps (l'_1 + l'_2)}\right)\\
 &\ge& 2^{-\eps l} \left(1-\frac{1}{2}\right) \left(2^{-\eps}\right)^l \left(1-\frac{1}{2}\right) \ge \frac{4^{-\eps l}}{4}
 \end{eqnarray}
 
 So, for LLL to work, the following have to be satisfied.
 $$
 \left(\frac{e}{\eps  \sqrt{|\Sigma|}}\right)^{\frac{\eps l}{1.01}} \le 
\frac{4^{-\eps l}}{4} \Leftrightarrow
4 \le \left(\frac{\eps  \sqrt{|\Sigma|}}{e4^{1.01}}\right)^{\frac{\eps l}{1.01}}
\Leftarrow 4\le \left(\frac{\eps  \sqrt{|\Sigma|}}{e4^{1.01}}\right)^{\frac{\eps\cdot 1/\eps^2}{1.01}} 
\Leftrightarrow \frac{4^{2.02(1+\eps)}e^2}{\eps^2} \le |\Sigma|
$$
 Therefore, for $|\Sigma| = \frac{4^{4	.04}e^2}{\eps^2}=O(\eps^{-2})$, the deterministic LLL conditions hold. This finishes the proof.
\end{proof}
}
\fullOnly{\ProofOfLLLConstruction}

\subsection{Boosting I: A Linear Time Construction of Synchronization Strings}\label{sec:BoostingStepII}

Next, we provide a simple boosting step which allows us to polynomially speed up any $\eps$-synchronization string construction. Essentially, we propose a way to construct an $O(\eps)$-synchronization string of length $O_\eps(n^2)$ having an $\eps$-synchronization string of length $n$. 

\begin{lemma}\label{lem:simplepolyboosting}
Fix an even $n \in \mathbb{N}$ and $\gamma > 0$ such that $\gamma n \in \mathbb{N}$. Suppose $S \in \Sigma^n$ is an $\eps$-synchronization string. The string $S' \in \Sigma'^{\gamma n^2}$ with $\Sigma' = \Sigma^3$ and 
\fullOnly{$$S'[i] = \left(S[i \bmod n], S[(i+n/2) \bmod n], S\left[\left\lceil {\frac{i}{\gamma n}}\right\rceil\right]\right)$$}
\shortOnly{$S'[i] = \left(S[i \bmod n], S[(i+n/2) \bmod n], S\left[\left\lceil {\frac{i}{\gamma n}}\right\rceil\right]\right)$}
 is an $(\eps + 6\gamma)$-synchronization string of length $\gamma n^2$. 
\end{lemma}
\global\def\ProofOfLemmaSimplepolyboosting{
\shortOnly{\subsection{Proof of Lemma~\ref{lem:simplepolyboosting}}}
\begin{proof}
Intervals of length at most $n/2$ lay completely within a copy of $S$ and thus have the $\eps$-synchronization property. For intervals of size $l$ larger than $n/2$ we look at the synchronization string which is blown up by repeating each symbol $\gamma n$ times. Ensuring that both sub-intervals contain complete blocks changes the edit distance by at most $3 \gamma n$ and thus by at most $6 \gamma l$. Once only complete blocks are contained we use the observation that the longest common subsequence of any two strings becomes exactly a factor $k$ larger if each symbols is repeated $k$ times in each string. This means that the relative edit distance does not change and is thus at least $\eps$. Overall this results in the $(\eps + 6\gamma)$-synchronization string property to hold for large intervals in $S'$.
\end{proof}
}
\fullOnly{\ProofOfLemmaSimplepolyboosting}
We use this step to speed up the polynomial time deterministic $\eps$-synchronization string construction in Theorem~\ref{thm:polyConstruction} to linear time.

\begin{theorem}\label{thm:linearTimeConstruction}
There exists an algorithm that, for any $0< \eps < 1$, constructs an $\eps$-synchronization string of length $n$ over an alphabet of size $\eps^{-O(1)}$ in $O(n)$ time.
\end{theorem}
\global\def\ProofOfThmLinearTimeConstruction{
\shortOnly{\subsection{Proof of Theorem~\ref{thm:linearTimeConstruction}}}
\begin{proof}
Note that if one takes an $\eps'$-synchronization strings of length $n'$ and applies the boosting step in Theorem~\ref{lem:simplepolyboosting} $k$ times with parameter $\gamma$, he would obtain a $(\eps'+6k\gamma)$-synchronization string of length $\gamma^{2^k-1}n^{2^k}$. 

For any $0<\eps<1$, Theorem~\ref{thm:polyConstruction} gives a deterministic algorithm for constructing an $\eps$-synchronization string over an alphabet $O(\eps^{-4})$ that takes $O(n^T)$ time for some constant $T$ independent of $\eps$ and $n$.
We use the algorithm in Theorem~\ref{thm:polyConstruction} to construct an $\eps'=\frac{\eps}{2}$ synchronization string of length $n'=\frac{n^{1/T}}{\gamma}$ for $\gamma=\frac{\eps}{12\log T}$ 
over an alphabet of size $O(\eps^{-4})$
in $O({n'}^{T}) = O(n)$ time. Then, we apply boosting step I $k=\log T$ times with $\gamma=\frac{\eps}{12\log T}$ to get an $(\eps'+6\gamma\log T = \eps)$-synchronization string of length $\gamma^{T-1}n'^T \ge n$. 
As boosting step have been employed constant times, the eventual alphabet size will be $\eps^{-O(1)}$ and the run time is $O(n)$.
\end{proof}
}
\fullOnly{\ProofOfThmLinearTimeConstruction}

\subsection{Boosting II: Explicit Constructions for Long-Distance Synchronization Strings}\label{sec:BoostingStepI}
We start this section by a discussion of \emph{explicitness} quality of synchronization string constructions. In addition to the time complexity of synchronization strings' constructions, an important quality of a construction that we take into consideration for applications that we will discuss later is explicitness or, in other words, how fast one can calculate a particular symbol of a synchronization string.

\begin{definition}[$T(n)$-explicit construction]
If a synchronization string construction algorithm can compute $i$th index of the string it is supposed to find, i.e., $S[i]$, in $T(n)$ we call it an \emph{$T(n)$-explicit} algorithm. 
\end{definition}

We are particularly interested in cases where $T(n)$ is polylogarithmically large in terms of $n$. For such $T(n)$, a $T(n)$-explicit construction implies a near-linear construction of the entire string as one can simply compute the string by finding out symbols one by one in $n\cdot T(n)$ overall time. We use the term \emph{highly-explicit} to refer to $O(\log n)$-explicit constructions.

We now introduce a boosting step in Lemma~\ref{thm:CodeBlock} that will lead to explicit constructions of (long-distance) synchronization strings. Lemma~\ref{thm:CodeBlock} shows that, using a high-distance insertion-deletion code, one can construct strings that satisfy the requirement of long-distance synchronization strings for every pair of substrings that are of total length $\Omega_\eps(\log n)$ or more. Having such a string, one can construct a $O_\eps(1)$-long-distance $\eps$-synchronization string by simply concatenating the outcome of Lemma~\ref{thm:CodeBlock} with repetitions of an $O_\eps(\log n)$-long $\eps$-synchronization string. 

This boosting step is deeply connected to our new definition of long-distance $\eps$-synchronization strings. In particular, we observe the following interesting connection between insertion-deletion codes and long-distance $\eps$-synchronization strings.

\begin{lemma}
If $S$ is a $c$-long-distance $\eps$-synchronization string where $c=\theta(1)$ then $\mathcal{C} =\{S(i \cdot c\log n, (i+1) \cdot c\log n] | 0 \leq i < \frac{n}{c\log n}-1\}$ is an insdel error correcting code with minimum distance at least $1 - \eps$ and constant rate. 
Further, if $S$ has a highly explicit construction, $\mathcal{C}$ has a linear encoding time.
\end{lemma}
\begin{proof}
The distance follows from the definition of long-distance $\eps$-synchronization strings. The rate follows because the rate $R$ is equal to $R = \frac{\log |\mathcal{C}|}{c \log n \log q} = \frac{\log \frac{n}{c\log n}}{O(\log n)} = \Omega(1)$. Finally, as 
$S$ is highly explicit and $\left|S(i \cdot c\log n, (i+1) \cdot c\log n]\right|  = c\log n$, one can compute $S(i \cdot c\log n, (i+1) \cdot c\log n]$ in linear time of its length which proves the linear construction.
\end{proof}

Our boosting step is mainly built on the converse of this observation. 

\begin{lemma}\label{thm:CodeBlock}
Suppose $\mathcal{C}$ is a block insdel code over alphabet of size $q$, block length $N$, distance $1-\eps$ and rate $R$ and let $S$ be a string obtained by attaching all codewords back to back in any order. Then, for $\eps'=4\eps$, $S$ is a string of length $n=q^{R\cdot N}\cdot N$ which satisfies the long-distance $\eps'$-synchronization property for any pair of intervals of aggregated length $\frac{4}{\eps} N
 \le \frac{4}{\eps\log q} (\log n-\log R)
$ or more. 
Further, if $\mathcal{C}$ is linearly encodable, $S$ has a highly explicit construction.
\end{lemma}
\shortOnly{
\begin{proof}[Proof Sketch]
To prove the claim, we essentially show that the edit distance of two strings consisting of a large enough number of codewords of $\mathcal{C}$ in total can be bounded below in terms of $\eps$.
\end{proof}
}
\global\def\ProofOfLemCodeBlock
{\shortOnly{\subsection{Proof of Lemma~\ref{thm:CodeBlock}}}
\begin{proof}
The length of $S$ follows from the definition of rate. Moreover, the highly explicitness follows from the fact that every substring of $S$ of length $\log n$ may include parts of 
$\frac{1}{\eps\log q} + 1$ codewords each of which can be computed in linear time in terms of their length. Therefore, any substring $S[i, i+\log n]$ can be constructed in $O\left(\max\left\{\frac{\log n}{\eps\log q}, \log n\right\}\right) = O_{\eps, q}(\log n)$. 
To prove the long distance property, we have to show that for every four indices $i_1<j_1\le i_2 < j_2$ where $j_1 + j_2 - i_1 - i_2 \ge \frac{4 N}{\eps}$, we have
\begin{equation}\label{eqn:claim}
ED(S[i_1, j_1),S[i_2, j_2))\ge(1-4\eps) (j_1 + j_2 - i_1 - i_2).
\end{equation}

Assume that $S[i_1, j_1)$ contains a total of $p$ complete blocks of $\mathcal{C}$ and $S[i_2, j_2)$  contains $q$ complete blocks of $\mathcal{C}$. Let $S[i'_1, j'_1)$ and $S[i'_2, j'_2)$ be the strings obtained be throwing the partial blocks away from $S[i_1, j_1)$ and $S[i_2, j_2)$. 
Note that the overall length of the partial blocks in $S[i_1, j_1)$ and $S[i_2, j_2)$ is less than $4N$, which is at most an $\eps$-fraction of $S[i_1, j_1)\cup S[i_2, j_2)$, since $\frac{4N}{4N/\eps} < \eps$.

Assume by contradiction that $ED(S[i_1, j_1), S[i_2, j_2)) < (1-4\eps) (j_1 + j_2 - i_1 - i_2)$. Since edit distance preserves the triangle inequality, we have that 
\begin{align*}
ED\left(S[i'_1, j'_1),S[i'_2, j'_2)\right) &\le ED\left(S[i_1, j_1),S[i_2, j_2)\right) + |S[i_1, i'_1)|+ |S[j'_1, j_1)| + |S[i_2, i'_2)|+ |S[j'_2, j_2)|\\ 
\leq &\left(1-4\eps\right)(j_1 + j_2 - i_1 - i_2) + \eps(j_1 + j_2 - i_1 - i_2)\\
\leq &\left(1-4\eps+\eps\right)(j_1 + j_2 - i_1 - i_2)\\
 <&\left(\frac{1-3\eps}{1-\eps}\right)\left((j'_1 - i'_1) + (j'_2 - i'_2)\right). 
 \end{align*}

This means that the longest common subsequence of $S[i'_1, j'_1)$ and $S[i'_2, j'_2)$ has length of at least \[\frac{1}{2}\left[\left(|S[i'_1, j'_1)|+|S[i'_2, j'_2)|\right) \left(1- \frac{1-3\eps}{1-\eps}\right)\right],\] which means that there exists a monotonically increasing matching between $S[i'_1, j'_1)$ and $S[i'_2, j'_2)$ of the same size.
Since the matching is monotone, there can be at most $p+q$ pairs of error-correcting code blocks having edges to each other. The Pigeonhole Principle implies that there are two error-correcting code blocks $B_1$ and $B_2$ such that the number of edges between them is at least

\begin{eqnarray*}
&&\frac{\frac{1}{2}\left[\left(|S[i_1, j_1)|+|S[i_2, j_2)|\right) \left(1- \frac{1-3\eps}{1-\eps}\right)\right]}{p+q}\\
 &=& \frac{(p+q)N\left(1-\frac{1-3\eps}{1-\eps}\right)}{2(p+q)}\\
&>&\frac{1}{2} \left(1-\frac{1-3\eps}{1-\eps}\right)\cdot N.
\end{eqnarray*}

 Notice that this is also a lower bound on the longest common subsequence of $B_1$ and $B_2$. This means that
$$ED(B_1, B_2) < 2N - \left(1-\frac{1-3\eps/4}{1-\eps/4}\right)N<\left(1+\frac{1-3\eps}{1-\eps}\right)N = \frac{2-4\eps}{1-\eps}N < 2\left(1-\eps\right)N.$$

This contradicts the error-correcting code's distance property, which we assumed to be larger than $2(1-\eps)N$, and therefore we may conclude that for all indices $i_1<j_1\le i_2 < j_2$ where $j_1 + j_2 - i_1 - i_2 \ge \frac{4N}{\eps}$, \eqref{eqn:claim} holds.
\end{proof}
}
\fullOnly{\ProofOfLemCodeBlock}

We point out that even a brute force enumeration of a good insdel code could be used to give an $\eps$-synchronization string for long distance. All is needed is a string for small intervals. This one could be brute forced as well. Overall this gives an alternative polynomial time construction (still using the inspiration of long-distance codes, though). More importantly, if we use a linear time construction for the short distances and a linear time encodable insdel code, we get a simple $O_\eps(\log n)$-explicit long-distance $\eps$-synchronization string construction for which any interval $[i, i + O_\eps(\log n)]$ is computable in $O_\eps(\log n)$.

In the rest of this section, as depicted in Figure~\ref{fig:flow}, we first introduce a high distance, small alphabet error correcting code that is encodable in linear time in Lemma~\ref{lem:highDistanceSmallAlphabetECC} using a high-distance linear-time code introduced in~\cite{guruswami2005linear}. We then turn this code into a high distance insertion deletion code using the indexing technique from~\cite{haeupler2017synchronization}. Finally, we will employ this insertion-deletion code in the setup of Lemma~\ref{thm:CodeBlock} to obtain a highly-explicit linear-time long-distance synchronization strings.

\global\def\InsdelCodeProofPrelude{
Our codes are based on the following code from Guruswami and Indyk~\cite{guruswami2005linear}.

\begin{theorem}[Theorem 3 from \cite{guruswami2005linear}]\label{thm:GuruswamiIndykHighRateCodes}
For every $r$, $0 < r < 1$, and all sufficiently small $\epsilon > 0$, there exists a family of codes of rate $r$ and relative distance at least $(1 - r - \epsilon)$ over an alphabet of size $2^{O(\epsilon^{-4}r^{-1}\log(1/\epsilon))}$ such that codes from the family can be encoded in linear time and can also be (uniquely) decoded in linear time from $2(1 - r - \epsilon)$ fraction of half-errors, i.e., a fraction $e$ of errors and $s$ of erasures provided $2e + s \le (1 - r - \epsilon)$.
\end{theorem}

One major downside of constructing $\eps$-synchronization strings based on the code from Theorem~\ref{thm:GuruswamiIndykHighRateCodes} is the exponentially large alphabet size in terms of $\eps$. We concatenate this code with an appropriate small alphabet code to obtain a high-distance code over a smaller alphabet size.

\begin{lemma}\label{lem:smallAlphabetRandomCode}
For sufficiently small $\eps$ and $A, R > 1$, and any set $\Sigma_i$ of size $|\Sigma_i|=2^{O(\eps^{-5}\log(1/\eps))}$, there exists a code $C:\Sigma_i\rightarrow\Sigma_o^N$ with distance $1-\eps$ and rate $\eps^{R}$ where $|\Sigma_o| = O(\eps^{-A})$.
\end{lemma}
\global\def\ProofOfLemSmallAlphabetRandomCode{
\begin{proof}
To prove the existence of such code, we show that a random code with distance $\delta= 1-\eps$, rate $r=\eps^{A}$, alphabet size $|\Sigma_o|=\eps^{-A}$, and block length 
$$N=\frac{\log |\Sigma_i|}{\log |\Sigma_o|}\cdot \frac{1}{r} = O\left(\frac{\eps^{-5}\log(1/\eps)}{A\log(1/\eps)}\cdot\frac{1}{\eps^{R}}\right) = 
\frac{1}{A}\cdot O\left(\eps^{-5 - R}\right)$$
exists with non-zero probability. The probability of two randomly selected codewords of length $N$ out of $\Sigma_o$ being closer than $\delta=1-\eps$ can be bounded above by the following term.
$${N\choose N\eps} \left(\frac{1}{|\Sigma_o|}\right)^{-N\eps}$$
Hence, the probability of the random code with $|\Sigma_o|^{Nr} = |\Sigma_1|$ codewords having a minimum distance smaller than $\delta=1-\eps$ is at most the following.
\begin{eqnarray*}
&&{N\choose N\eps} \left(\frac{1}{|\Sigma_o|}\right)^{N\eps}{|\Sigma_i|\choose 2}\\
&\le&\left(\frac{N e}{N\eps}\right)^{N\eps} \frac{|\Sigma_i|^{2}}{|\Sigma_o|^{N\eps}}\\
&=& \left(\frac{e}{\eps}\right)^{N\eps} \frac{ 2^{O(\eps^{-5}\log(1/\eps))}}{(\eps^{-A})^{N\eps}}\\
&=& 2^{O((1-A)\log(1/\eps)N\eps + \eps^{-5}\log(1/\eps))}\\
&=& 2^{(1-A)O(\eps^{-4-R}\log(1/\eps)) + O(\eps^{-5}\log(1/\eps))}
\end{eqnarray*}

For $A>1$, $1-A$ is negative and for $R > 1$, $\eps^{-4-R}\log(1/\eps)$ is asymptotically larger than $\eps^{-5}\log(1/\eps)$. Therefore, for sufficiently small $\eps$, the exponent is negative and the desired code exists.
\end{proof}
}
\ProofOfLemSmallAlphabetRandomCode

Concatenating the code from Theorem~\ref{thm:GuruswamiIndykHighRateCodes} (as the outer code) and the code from Lemma~\ref{lem:smallAlphabetRandomCode} (as inner code) gives the following code.

\begin{lemma}\label{lem:highDistanceSmallAlphabetECC}
For sufficiently small $\eps$ and any constant $0<\gamma$, there exists an error correcting code of rate $O(\eps^{2.01})$ and distance $1-\eps$ over an alphabet of size $O(\eps^{-(1+\gamma)})$ which is encodable in linear time and also uniquely decodable from an $e$ fraction of erasures and $s$ fraction of symbol substitutions when $s+2e < 1-\eps$ in linear time.
\end{lemma}
\global\def\ProofOfLemHighDistanceSmallAlphabetECC{
\begin{proof}
To construct such code, we simply codes from Theorem~\ref{thm:GuruswamiIndykHighRateCodes} and Lemma~\ref{lem:smallAlphabetRandomCode} as outer and inner code respectively. Let $\mathcal{C}_1$ be an instantiation of the code from Theorem~\ref{thm:GuruswamiIndykHighRateCodes} with parameters $r=\eps/4$ and $\epsilon = \eps/4$. Code $\mathcal{C}_1$ is a code of rate $r_1 = \eps/4$ and distance $\delta_1=1-\eps/4-\eps/4 = 1-\eps/2$ over an alphabet $\Sigma_1$ of size 
$2^{O(\epsilon^{-4}r^{-1}\log(1/\epsilon))} = 2^{O(\eps^{-5}\log(1/\eps))}$
 which is encodable and decodable in linear time. 

Further, according to Lemma~\ref{lem:smallAlphabetRandomCode}, one can find a code $\mathcal{C}_2:\Sigma_1 \rightarrow \Sigma_2^{N_2}$ for $\Sigma_2=\eps^{-(1+\gamma)}$ with distance $\delta_2=1-\eps/2$ rate $r_2=O(\eps^{1.01})$ by performing a brute-force search. Note that block length and alphabet size of $\mathcal{C}_2$ is constant in terms of $n$. Therefore, such code can be found in $O_\eps(1)$ and by forming a look-up table can be encoded and decoded from $\delta$ half-errors in $O(1)$. Hence, concatenating codes $\mathcal{C}_1$ and $\mathcal{C}_2$ gives a code of distance 
$\delta=\delta_1 \cdot\delta_2=(1-\eps/2)^2 \ge 1-\eps$ and rate $r=r_1\cdot r_2 = O(\eps^{2.01})$ over an alphabet of size $|\Sigma_2|=O\left(\eps^{-(1+\gamma)}\right)$ which can be encoded in linear time in terms of block length and decoded from $e$ fraction of erasures and $s$ fraction of symbol substitutions when $s+2e < 1-\eps$ in linear time as well.
\end{proof}
}
\ProofOfLemHighDistanceSmallAlphabetECC

Indexing the codewords of a code from Lemma~\ref{lem:highDistanceSmallAlphabetECC} with linear-time constructible synchronization strings of Theorem~\ref{thm:linearTimeConstruction} using the technique from~\cite{haeupler2017synchronization} summarized in Theorem~\ref{thm:mainECC} gives Theorem~\ref{thm:insdelCode}.
}\fullOnly{\InsdelCodeProofPrelude}
\begin{theorem}\label{thm:insdelCode}
For sufficiently small $\eps$, there exists a family of insertion-deletion codes with rate $\eps^{O(1)}$ that correct from $1-\eps$ fraction of insertions and deletions over an alphabet of size $\eps^{O(1)}$ that is encodable in linear time and decodable in quadratic time in terms of the block length.
\end{theorem}
\global\def\ProofOfTheoremInsdelCode{
\shortOnly{\subsection{Proof of Theorem~\ref{thm:insdelCode}}
\subsubsection{Some useful Lemmas}
\InsdelCodeProofPrelude
\subsubsection{Proof of Theorem~\ref{thm:insdelCode}}}
\begin{proof}
Theorem~\ref{thm:mainECC} provides a technique to convert an error correcting code into an insertion-deletion code by indexing the codewords with a synchronization string. We use the error correcting code $\mathcal{C}$ from Lemma~\ref{lem:highDistanceSmallAlphabetECC} with parameter $\eps'=\eps/2$ and $\gamma=0.01$ along with a linear-time constructible synchronization strings $S$ from Theorem~\ref{thm:linearTimeConstruction} with parameter $\eps''=(\eps/2)^2$ in the context of Theorem~\ref{thm:mainECC}. We also use the global decoding algorithm from Haeupler and Shahrasbi~\cite{haeupler2017synchronization} for the synchronization string. This will give an insertion deletion code over an alphabet of size $\eps^{O(1)}$ corrects from $(1-\eps') -\sqrt{\eps''}=1-\eps$ insdels with a rate of
\fullOnly{$$\frac{r_{\mathcal{C}}}{1+|\Sigma_{S}|/|\Sigma_{\mathcal{C}}|}
=\frac{O\left(\eps^{2.01}\right)}{1+O(\eps''^{-O(1)}/\eps^{-1.01})} = \eps^{O(1)}.$$}
\shortOnly{$\frac{r_{\mathcal{C}}}{1+|\Sigma_{S}|/|\Sigma_{\mathcal{C}}|}
=\frac{O\left(\eps^{2.01}\right)}{1+O(\eps''^{-O(1)}/\eps^{-1.01})} = \eps^{O(1)}.$}
As $\mathcal{C}$ is encodable and $S$ is constructible in linear time, the encoding time for the insdel code will be linear. Further, as $\mathcal{C}$ is decodable in linear time and $S$ is decodable in quadratic time (using global decoding from~\cite{haeupler2017synchronization}), the code is decodable in quadratic time.
\end{proof}
}
\fullOnly{\ProofOfTheoremInsdelCode}
Using insertion-deletion code from Theorem~\ref{thm:insdelCode} and boosting step from Lemma~\ref{thm:CodeBlock}, we can now proceed to the main theorem of this section that provides a highly explicit construction for $c=O_\eps(1)$-long-distance synchronization strings.

\begin{theorem}\label{thm:FiniteLongDistConstructionNEW}
There is a deterministic algorithm that, for any constant $0< \eps  < 1$ and $n \in \mathbb{N}$, computes an $c=\eps^{-O(1)}$-long-distance $\eps$-synchronization string $S \in \Sigma^n$ where $|\Sigma|=\eps^{-O(1)}$. Moreover, this construction is $O(\log n)$-explicit and can even compute $S[i, i+\log n]$ in $O_\eps(\log n)$ time.
\end{theorem}
\shortOnly{
\begin{proof}[Proof Sketch]
We construct such string by attaching codewords of a code from Theorem~\ref{thm:insdelCode} and attaching it, symbol by symbol, to a string $T$ that satisfies synchronization string property for (logarithmically) short intervals.
\begin{eqnarray}
R[i]=(S[i], T[i]) = \Bigg(\mathcal{C}\left(\left\lfloor \frac{i}{N}\right\rfloor\right) \left[i \left( \bmod\, N\right)\right], T[i]\Bigg)\label{eq:HighlyExplicitLongDistConstruction}
\end{eqnarray}
Such string $T$ can be found by attaching instances of a short synchronization string as is described in details in the full proof in Appendix~\ref{app:ProofOfThmFiniteLongDistConstructionNEW} and depicted in Figure~\ref{fig:construction}. 
\end{proof}
}
\global\def\ProofOfThmFiniteLongDistConstructionNEW{
\shortOnly{\subsection{Proof of Theorem~\ref{thm:FiniteLongDistConstructionNEW}}\label{app:ProofOfThmFiniteLongDistConstructionNEW}}
\begin{proof}
We simply use an insertion-deletion code from Theorem~\ref{thm:insdelCode} with parameter $\eps'=\eps/4$ and block length 
$N=\frac{\log_q n}{R}$ 
where $q=\eps^{-O(1)}$ is the size of the alphabet from Theorem~\ref{thm:insdelCode}. 
Using this code in Lemma~\ref{thm:CodeBlock} gives a string $S$ of length $q^{RN}\cdot N \ge n$ that satisfies $4\eps'=\eps$-synchronization property over any pair of intervals of total length 
$\frac{4N}{\eps} = O\left(\frac{\log n}{\eps R \log q}\right) = O\left(\eps^{-O(1)} \log n\right)$
or more. Since the insertion-deletion code from Theorem~\ref{thm:insdelCode} is linearly encodable, the construction will be highly-explicit.

To turn $S$ into a $c$-long-distance $\eps$-synchronization string for $c=\frac{4N}{\eps\log n}=O\left(\eps^{-O(1)}\right)$, we simply concatenate it with a string $T$ that satisfies $\eps$-synchronization property for neighboring intervals of total size smaller than $c\log n$. In other words, we propose the following structure for constructing $c$-long-distance $\eps$-synchronization string $R$.
\fullOnly{
\begin{eqnarray}
R[i]=(S[i], T[i]) = \Bigg(\mathcal{C}\left(\left\lfloor \frac{i}{N}\right\rfloor\right) \left[i \left( \bmod\, N\right)\right], T[i]\Bigg)\label{eq:HighlyExplicitLongDistConstruction}
\end{eqnarray}
}\shortOnly{
\begin{eqnarray*}
R[i]=(S[i], T[i]) = \Bigg(\mathcal{C}\left(\left\lfloor \frac{i}{N}\right\rfloor\right) \left[i \left( \bmod\, N\right)\right], T[i]\Bigg)
\end{eqnarray*}
}

Let $S'$ be an $\eps$-synchronization string of length $2c\log n$. Using linear-time construction from Theorem~\ref{thm:linearTimeConstruction}, one can find $S'$ in linear time in its length, i.e, $O(\log n)$. We define strings $T_1$ and $T_2$ consisting of repetitions of $S'$\fullOnly{ as follows.
$$T_1 = (S', S', \cdots, S'), \qquad T_2=(0^{c\log n}, S', S', \cdots, S')$$}
\shortOnly{as $T_1 = (S', S', \cdots, S'), T_2=(0^{c\log n}, S', S', \cdots, S')$.}
The string $T_1\cdot T_2$ satisfies $\eps$-synchronization strings for neighboring intervals of total length $c\log n$ or less as any such substring falls into one copy of $S'$. 
Note that having $S'$ one can find any symbol of $T$ in linear time. Hence, $T$ has a highly-explicit linear time construction. Therefore, concatenating $S$ and $T$ gives a linear time construction for $c$-long-distance $\eps$-synchronization strings over an alphabet of size $\eps^{-O(1)}$ that is highly-explicit and, further, allows computing any substring $[i, i+\log n]$ in $O(\log n)$ time. A schematic representation of this construction can be found in Figure~\ref{fig:construction}.
\begin{figure}
\centering
\includegraphics[scale=.90]{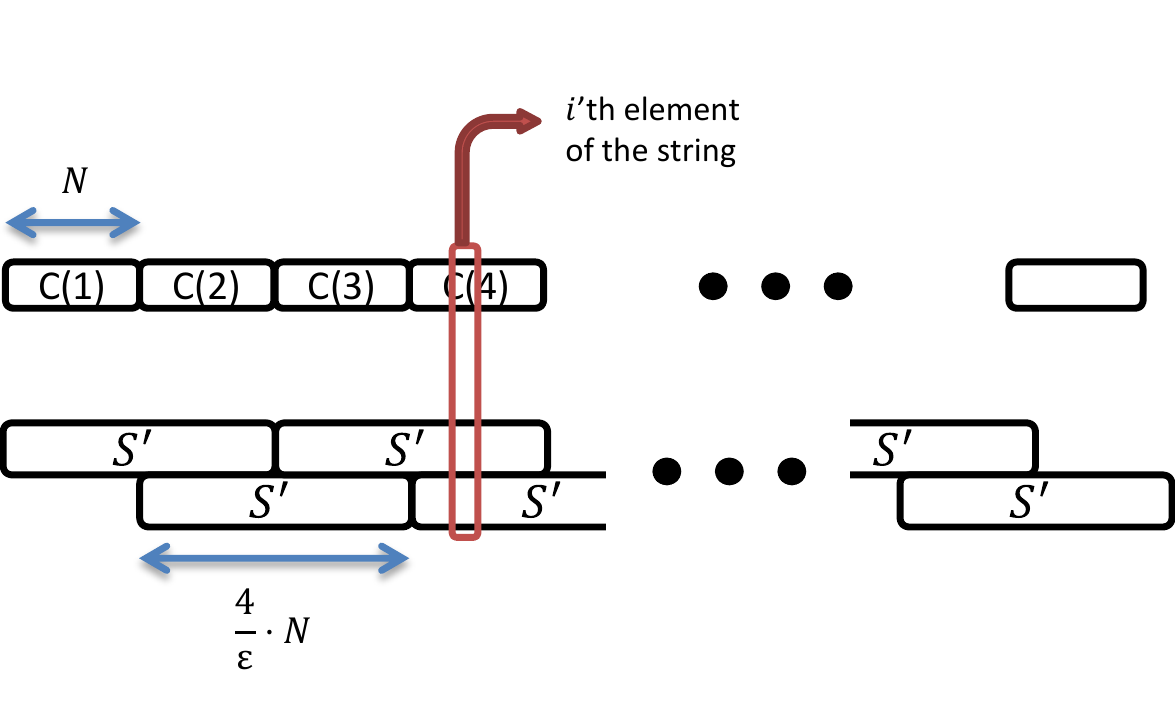}
\caption{Pictorial representation of the construction of a long-distance $\eps$-synchronization string of length $n$.}
\label{fig:construction}
\end{figure}
\end{proof}
}
\fullOnly{\ProofOfThmFiniteLongDistConstructionNEW}


\subsection{Infinite Synchronization Strings: Highly Explicit Construction}\label{sec:infiniteConstruction}
\fullOnly{Throughout this section we focus on construction of infinite synchronization strings. 
To measure the efficiency of a an infinite string's construction, we consider the required time complexity for computing the first $n$ elements of that string. Moreover, besides the time complexity, we employ a generalized notion of explicitness to measure the quality of infinite string constructions.}

In a similar fashion to finite strings, an infinite synchronization string is called to have a \emph{$T(n)$-explicit} construction if there is an algorithm that computes any position $S[i]$ in $O\left(T(i)\right)$. Moreover, it is said to have a highly-explicit construction if $T(i) = O(\log i)$.

We show how to deterministically construct an infinitely-long $\eps$-synchronization string over an alphabet $\Sigma$ which is polynomially large in $\eps^{-1}$. Our construction can compute the first $n$ elements of the infinite string in $O(n)$ time, is highly-explicit, and, further, can compute any $[i, i+\log i]$ in $O(\log i)$.

\begin{theorem}\label{thm:infiniteSync}
For all $0 < \eps < 1$, there exists an infinite $\eps$-synchronization string construction over a $poly(\eps^{-1})$-sized alphabet that is highly-explicit and also is able to compute $S[i, i+\log i]$ in $O(\log i)$. Consequently, using this construction, the first $n$ symbols of the string can be computed in $O(n)$ time.
\end{theorem}
\shortOnly{
\begin{proof}[Proof Sketch]
We construct such string using symbol by symbol concatenation of two strings that consist of finite synchronization strings with exponentially growing lengths (as shown in Figure~\ref{fig:constructive}) and prove that such string satisfies synchronization string property.
\end{proof}
}
\begin{figure}
\centering
\includegraphics[scale=.65]{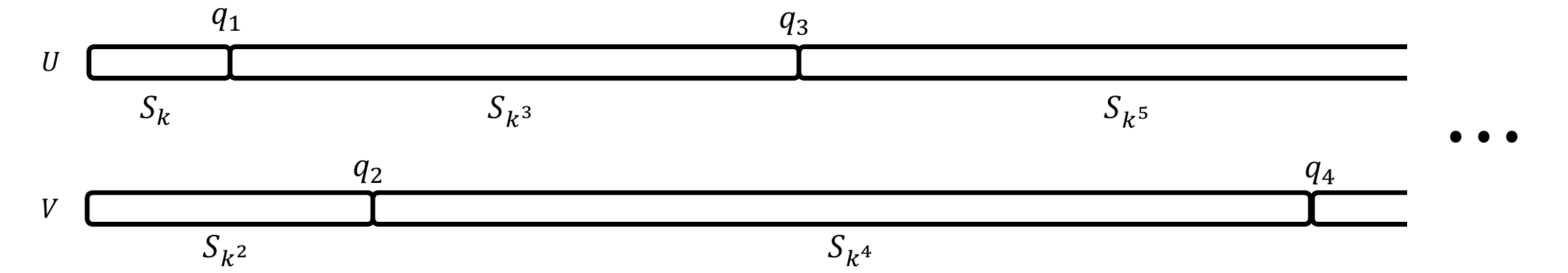}
\caption{Construction of Infinite synchronization string $T$}
\label{fig:constructive}
\end{figure}
\global\def\ProofOfThmInfiniteSync{
\shortOnly{\subsection{Proof of Theorem~\ref{thm:infiniteSync}}}
\begin{proof}
Let $k=\frac{6}{\eps}$ and let $S_i$ denote a $\frac{\eps}{2}$-synchronization string of length $i$. We define $U$ and $V$ as\fullOnly{ follows:
$$U = (S_k, S_{k^3}, S_{k^5}, \dots), \qquad V = (S_{k^2}, S_{k^4}, S_{k^6}, \dots)$$}\shortOnly{ $U = (S_k, S_{k^3}, S_{k^5}, \dots), V = (S_{k^2}, S_{k^4}, S_{k^6}, \dots)$.}
In other words, $U$ is the concatenation of $\frac{\eps}{2}$-synchronization strings of length $k, k^3, k^5, \dots$ and $V$ is the concatenation of $\frac{\eps}{2}$-synchronization strings of length $k^2, k^4, k^6, \dots$.
We build an infinite string $T$ such that $T[i] = (U[i], V[i])$ (see Figure~\ref{fig:constructive}).

First, if finite synchronization strings $S_{k^l}$ used above are constructed using the highly-explicit construction algorithm introduced in Theorem~\ref{thm:FiniteLongDistConstructionNEW}, any index $i$ can be computed by simply finding one index in two of $S_{k^l}$s in $O(\log n)$.  Further, any substring of length $n$ of this construction can be computed by constructing finite synchronization strings of total length $O(n)$. According to Theorem~\ref{thm:FiniteLongDistConstructionNEW}, that can be done in $O_\eps(n)$.



Now, all that remains is to show that $T$ is an $\eps$-synchronization string. We use following lemma to prove this.

\begin{lemma}\label{lem:covering}
Let $x < y < z$ be positive integers and let $t$ be such that $k^t \le |T[x, z)| < k^{t+1}$. Then there exists a block of $S_{k^i}$ in $U$ or $V$ such that all but a $\frac{3}{k}$ fraction of $T[x, z)$ is covered by $S_{k^i}$.
\end{lemma}

Note that this lemma shows that $ED(T[x, y), T[y, z)) > \left(1-\frac{\eps}{2}\right) \left(|T[x, y)| + |T[y, z)|\right) \left(1-\frac{3}{k}\right) = \left(1-\frac{\eps}{2}\right)^2 \left(|T[x, y)| + |T[y, z)|\right) \ge (1-\eps)\left(|T[x, y)| + |T[y, z)|\right)$, which implies that $T$ is an $\eps$-synchronization string. 
\end{proof}
\begin{proof}[Proof of Lemma \ref{lem:covering}]
We first define $i^{th}$ \emph{turning point} $q_i$ to be the index of $T$ at which $S_{k^{i+1}}$ starts, i.e., $q_i = k^{i} + k^{i - 2} + k^{i-4} + \cdots$. 
Note that 
\begin{eqnarray}
q_i &=& \Bigg\{
\begin{tabular}{cc}
$k^2+k^4+\cdots+k^i$ & Even $i$ \\[2mm]
$k+k^3+\cdots+k^i$ & Odd $i$
\end{tabular}\\
&=& \Bigg\{
\begin{tabular}{cc}
$k^2\frac{k^i-1}{k^2-1}$ & Even $i$ \\[2mm]
$k\frac{k^{i+1}-1}{k^2-1}$ & Odd $i$
\end{tabular}
\end{eqnarray}
Note that $q_{t-1} < 2 k^{t-1}$ and $|T[x, z)| \ge k^t$. Therefore, one can throw away all the elements of $T[x, z)$ whose indices are less than $q_{t - 1}$ without losing more than a $\frac{2}{k}$ fraction of the elements of $T[x, z)$. We will refer to the remaining part of $T[x, z)$ as $\tilde T$. 

Now, the distance of any two turning points $q_i$ and $q_j$ where $t \le i < j$ is at least $q_{t+1} - q_{t}$, and
\begin{eqnarray}
q_{t+1} - q_{t} 
&=& \Bigg\{
\begin{tabular}{cc}
$k\frac{k^{t+2}-1}{k^2-1} - k^2\frac{k^t-1}{k^2-1}$ & Even $t$ \\[2mm]
$k^2\frac{k^{t+1}-1}{k^2-1} - k\frac{k^{t+1}-1}{k^2-1}$ & Odd $t$
\end{tabular}\\
&=& \Bigg\{
\begin{tabular}{cc}
$\frac{(k-1)(k^{t+2}+k)}{k^2-1} = \frac{k^{t+2}+k}{k+1}$ & Even $t$ \\[2mm]
$\frac{(k-1)(k^{t+2}-k)}{k^2-1} = \frac{k^{t+2}-k}{k+1}$ & Odd $t$.
\end{tabular}
\end{eqnarray}

Hence, $q_{t+1} - q_{t} > k^{t+1} \left(1-\frac{1}{k}\right)$. Since $|\tilde{T}| \le |T[x, z)| < k^{t+1}$, this fact gives that there exists a $S_{k^i}$ which covers a $ \left(1-\frac{1}{k}\right)$ fraction of $\tilde T$. This completes the proof of the lemma.
\end{proof}
\fullOnly{A similar discussion for infinite long-distance synchronization string can be found in Appendix~\ref{sec:InfiniteLongDistanceConstruction}.}
}\fullOnly{\ProofOfThmInfiniteSync}

\global\def\SecInfiniteLongDistConstruction
{\section{Infinite long-Distance Synchronization Strings: Efficient Constructions}\label{sec:InfiniteLongDistanceConstruction}
In this section, we introduce and discuss the construction of infinite long-distance synchronization strings. The definition of $c$-long-distance $\eps$-synchronization  property strongly depends on the length of the string. This definition requires any two neighboring intervals as well as any two intervals of aggregated length of $c\log n$ or more to satisfy $\eps$-synchronization property. A natural generalization of this property to infinite strings would be to require similar guarantee to hold over all prefixes of it.

\begin{definition}[Infinite Long-Distance Synchronization Strings]
An infinite string $S$ is called a \emph{$c$-long-distance $\eps$-synchronization string} if any prefix of $S$ like $S[1, n]$ is a $c$-long-distance $\eps$-synchronization string of length $n$.
\end{definition}

%
We prove infinite long distance synchronization strings exist and provide efficient constructions for them. We prove this by providing a structure similar to the one proposed in Theorem~\ref{thm:infiniteSync} that constructed an infinite $\eps$-synchronization string using finite $\eps$-synchronization strings.

\begin{lemma}\label{lem:blackBoxLongDistInfiniteConstruction}
Let $\mathcal{A}(n)$ be an algorithm that computes a $c$-long-distance $\eps$-synchronization string $S\in\Sigma^n$ in $T(n)$ time. Further, let $\mathcal{A}_p(n, i)$ be an algorithm that computes $i$th position of a $c$-long-distance $\eps$-synchronization string of length $n$ in $T_p(n)$. Then, for any integer number $m\ge 2$, one can compose algorithms $\mathcal{A}'(n)$ and $\mathcal{A}'_p(i)$ that compute $S'[1, n]$ and $S'[i]$ respectively where $S'$ is an infinite $c$-long-distance $\left(\eps+\frac{4}{c\log m}\right)$-synchronization string over $\Sigma \times \Sigma$. Further, $\mathcal{A}'(n)$ and $\mathcal{A}'_p(i)$ run in 
$\min\left\{T(m^n), n\cdot T_p(m^n)\right\}$
 and $T_p(m^i)$ time respectively.
\end{lemma}
\begin{proof}
We closely follow the steps we took in Theorem~\ref{thm:infiniteSync}, except, instead of using geometrically increasing synchronization strings in construction of $U$ and $V$, we will use $c$-long-distance $\eps$-synchronization strings whose length increase in the form of a tower function. We define the tower function $tower(p, i)$ for $p\in \mathbb{R}, i \in \mathbb{Z}^+$ recursively as follows: Let $tower(p, 1) = p$ and for $i > 1$, $tower(p, i) = p^{tower(p, i-1)}$. Then, we define two infinite strings $U$ and $V$ as follows:
$$U = (S_m, S_{m^{m^m}}, \dots), \qquad V = (S_{m^m}, S_{m^{m^{m^m}}}, \dots).$$
where $S_l$ is a $c$-long-distance $\eps$-synchronization string of length $l$. We define the infinite string $T$ as the point by point concatenation of $U$ and $V$. 

We now show that this string satisfies the $c$-long-distance $\left(\eps+\frac{4}{c\log m}\right)$-synchronization property. 
We define \emph{turning points} $\{q_i\}_{i=1}^\infty$ in the same manner as we did in Theorem~\ref{thm:infiniteSync}, i.e., the indices of $T$ where a $S_{tower(m, i)}$ starts. Let $q_i$ be the index where $S_{tower(m, i+1)}$ starts.

Consider two intervals $[i_1, j_1)$ and $[i_2, j_2)$ where $j_1 \le i_2$ and $(j_1 - i_1) + (j_2 - i_2) \ge c\log j_2$. Let $k$ be an integer for which $q_k<j_2\le q_{k+1}$. 
Then, $(j_1 - i_1) + (j_2 - i_2) \ge c\log j_2 \ge c \log\left(tower(m, k)\right) = c\log m 
\cdot tower(m, k-1)$.
Note that all but $tower(m, k-1) + tower(m, k-3)+\cdots \le 2\cdot tower(m, k-1)$ many elements of $T[i_1, j_1) \cup T[i_2, j_2)$ lay in $T[q_{k-1}, q_{k+1})$ which is covered by $S_{tower(m, k-1)}$. Therefore, for $l = (j_1-i_1) + (j_2-i_2)$
\begin{eqnarray*}
ED(T[i_1, j_1), T[i_2, j_2)) &\ge&  ED(T[\max\{i_1, q_{k-1}\}, j_1), T[i_2, j_2)) - 2\cdot tower(m, k-1)\\
&\ge& (1-\eps)\cdot[(j_2-i_2) + (j_1-\max\{i_1, q_{k-1}\})] - 2\cdot tower(m, k-1)\\
&\ge& (1-\eps)\cdot[l-2\cdot tower(m, k-1)] - 2\cdot tower(m, k-1)\\
&\ge& (1-\eps)\cdot l - 4\cdot tower(m, k-1)\\
&\ge& \left(1-\eps-\frac{4\cdot tower(m, k-1)}{l}\right)\cdot l\\
&\ge& \left(1-\eps-\frac{4}{c\log m}\right)\cdot l
\end{eqnarray*}

Further, any two neighboring intervals $[i_1, i_2)$ and $[i_2, i_3)$ where $i_3 - i_1 < c \log i_3$ and $k\le i_3<k+1$, $[i_1, i_3)$ completely lies in $S_{k-1}$ and therefore $\eps$-synchronization property for short neighboring intervals holds as well. Thus, this string satisfies infinite $c$-long-distance $\left(\eps+\frac{4}{c\log m}\right)$-synchronization property.

Finally, to compute index $i$ of infinite string $T$ constructed as mentioned above, one needs to compute a single index of two finite $c$-long-distance $\eps$-synchronization strings of length $m^i$ or less. Therefore, computing $T[i]$ takes $T_p(m^i)$. This also implies that $T[1, n]$ can be computed in $n\cdot T_p(m^n)$. Clearly, on can also compute $T[1, n]$ by computing all finite strings that appear within the first $n$ elements. Hence, $T[1, n]$ is computable in $\min\left\{T(m^n), n\cdot T_p(m^n)\right\}$.
\end{proof}

Utilizing the construction proposed in Lemma~\ref{lem:blackBoxLongDistInfiniteConstruction} with $m = 2$ along with the highly-explicit finite $O_\eps(1)$-long-distance $\frac{\eps}{2}$-synchronization string construction introduced in Theorem~\ref{thm:FiniteLongDistConstructionNEW}, results in the following infinite string construction:

\begin{theorem}
For any constant $0 < \eps  < 1$ there is a deterministic algorithm which computes $i$th position of an infinite $c$-long-distance $\eps$-synchronization string $S$ over an alphabet of size $|\Sigma|=\eps^{-O(1)}$ where $c=O_\eps(1)$ in $O_\eps(i)$ time. This implies a quadratic time construction for any prefix of such string.
\end{theorem}
}

\section{Local Decoding}\label{sec:localDecoding}

In Section~\ref{sec:sync_construction}, we discussed the close relationship between long-distance synchronization strings and insdel codes and provided highly-explicit constructions of long-distance synchronization strings based on insdel codes.

In this section, we make a slight modification to the highly explicit structure~\eqref{eq:HighlyExplicitLongDistConstruction} we introduced in Theorem~\ref{thm:FiniteLongDistConstructionNEW} where we showed one can use a constant rate insertion-deletion code $\mathcal{C}$ with distance $1-\frac{\eps}{4}$ and block length $N=O(\log n)$ and a string $T$ satisfying $\eps$-synchronization property for pairs of neighboring intervals of total length $c\log n$ or less to make a $c$-long-distance synchronization string of length $n$. 
In addition to the symbols of the string consisting of codewords of $\mathcal{C}$ and symbols of string $T$, we append $\Theta\left(\log \frac{1}{\eps}\right)$ extra bits to each symbol to enable \emph{local decodability}. This extra symbol, as described in~\eqref{eq:LocalDecConstruction}, essentially works as a circular index counter for insertion-deletion code blocks. 
\begin{eqnarray}
R[i]=\Bigg(\mathcal{C}\left(\left\lfloor \frac{i}{N}\right\rfloor\right) \left[i \left( \bmod\, N\right)\right], T[i], \left\lfloor \frac{i}{N}\right\rfloor \left(\bmod\, \frac{8}{\eps^3}\right)\Bigg)\label{eq:LocalDecConstruction}
\end{eqnarray}
With this extra information appended to the construction, we claim that \emph{relative suffix error density} is smaller than $\eps$ upon arrival of some symbol, then one can decode the corresponding index correctly by only looking at the last $O(\log n)$ symbols. At any point of a communication over an insertion-deletion channel, relative suffix error density is defined as maximum fraction of errors occurred over all suffixes of the message sent so far. (see Definition 5.12 from~\cite{haeupler2017synchronization}).

\begin{theorem}\label{thm:LocalDecodingFiniteModular}
Let $R$ be a highly-explicit long-distance $\eps$-synchronization string constructed according to~\eqref{eq:LocalDecConstruction}. Let $R[1, i]$ be sent by Alice and be received as $R'[1, j]$ by Bob. If relative suffix error density is smaller than $1-\frac{\eps}{2}$, then Bob can find $i$ in 
$\frac{4}{\eps}\cdot T_{Dec}(N) +\frac{4N}{\eps}\cdot(T_{Enc}(N) + Ex_T(c\log n) + c^2\log^2n)$ only by looking at the last $\max(\frac{4N}{\eps^2}, c\log n)$ received symbols where $T_{Enc}$ and $T_{Dec}$ is the encoding and decoding complexities of $\mathcal{C}$ and $Ex_T(l)$ is the amount of time it takes to construct a substring of $T$ of length $l$.
\end{theorem}

For linear-time encodable, quadratic-time decodable code $\mathcal{C}$ and highly-explicit string $T$ constructed by repetitions of short synchronization strings used in Theorem~\ref{thm:FiniteLongDistConstructionNEW}, construction~\eqref{eq:LocalDecConstruction} provides the following.

\begin{theorem}\label{thm:LocalDecodingFinite}
Let $R$ be a highly-explicit long-distance $\eps$-synchronization string constructed according to~\eqref{eq:LocalDecConstruction} with code $\mathcal{C}$ and string $T$ as described in Theorem~\ref{thm:FiniteLongDistConstructionNEW}. Let $R[1, i]$ be sent by Alice and be received as $R'[1, j]$ by Bob. If relative suffix error density is smaller than $1-\frac{\eps}{2}$, then Bob can find $i$ in $O(\log^3 n)$ only by looking at the last $O(\log n)$ received symbols.
\end{theorem}

This decoding procedure, which we will refer to as \emph{local decoding} consists of two principal phases upon arrival of each symbol. During the first phase, the receiver finds a list of $\frac{1}{\eps}$ numbers that is guaranteed to contain the index of the current insertion-deletion code block. This gives $\frac{N}{\eps}$ candidates for the index of the received symbol. The second phase uses the relative suffix error density guarantee to choose the correct candidate among the list. The following lemma formally presents the first phase. This idea of using list decoding as a middle step to achieve unique decoding has been used by several previous work~\cite{GL-isit16, GW-random15, ghaffari2014optimal, guruswami2006explicit}.

\begin{lemma}\label{lem:fastDecodingPhaseOne}
Let $S$ be an $\eps$-synchronization string constructed as described in~\eqref{eq:LocalDecConstruction}. Let $S[1, i]$ be sent by Alice and be received as $S_\tau[1, j]$ by Bob. If relative suffix error density is smaller than $1-\eps/2$, then Bob can compute a list of $\frac{4N}{\eps}$ numbers that is guaranteed to contain $i$.
\end{lemma}
\begin{proof}
Note that as relative suffix error density is smaller than $1-\eps/2 < 1$, the last received symbol has to be successfully transmitted. Therefore, Bob can correctly figure out the insertion-deletion code block index counter value which we denote by $count$. 
Note that if there are no errors, all symbols in blocks with index counter value of $count, count -1, \cdots, count -4/\eps+1 \bmod \frac{8}{\eps^3}$ that was sent by Bob right before the current symbol, have to be arrived within the past $4/\eps\cdot N$ symbols. However, as adversary can insert symbols, those symbols can appear anywhere within the last $\frac{2}{\eps}\frac{4N}{\eps}=\frac{8N}{\eps^2}$ symbols. 

Hence, if Bob looks at the symbols arrived with index $i\in\{count, count -1, \cdots, count -4/\eps+1\}  \bmod \frac{8}{\eps^3}$ within the last $\frac{8N}{\eps^2}$ received symbols, he can observe all symbols coming from blocks with index $count, count -1, \cdots, count -4/\eps+1 \bmod \frac{8}{\eps^3}$ that was sent right before $S[i]$. Further, as our counter counts modulo $\frac{8}{\eps^3}$, no symbols from older blocks with indices $count, count -1, \cdots, count -1/\eps+1 \bmod \frac{4}{\eps^3}$ will appear within the past $\frac{8N}{\eps^2}$ symbols. Therefore, Bob can find the symbols from the last $\frac{4}{\eps}$ blocks up to some insdel errors. By decoding those blocks, he can make up a list of $\frac{4}{\eps}$ candidates for the actual block number. As each block contains $N$ elements, there are a total of $\frac{4N}{\eps}$ many candidates for $i$.

Note that as relative suffix error density is at most $1-\eps/2$ and the last block may not have been completely sent yet, the total fraction of insdels in reconstruction of the last $\frac{4}{\eps}$ blocks on Bob side smaller than $1-\eps/2 + \frac{N}{4N/\eps^2} \le 1-\frac{\eps}{4}$. Therefore, the error density in at least one of those blocks is not larger than $1-\frac{\eps}{4}$. This guarantees that at least one block will be correctly decoded and henceforth the list contains the correct actual index.
\end{proof}
 
We now define a limited version of relative suffix distance (defined in~\cite{haeupler2017synchronization}) which enables us to find the correct index among candidates found in Lemma~\ref{lem:fastDecodingPhaseOne}.
 
\begin{definition}[Limited Relative Suffix Distance]
For any two strings $S, S' \in \Sigma^*$ we define their $l$-limited relative suffix distance\fullOnly{, $l-LRSD$, as follows:
$$l-LRSD(S,S') = \max_{0 < k < l} \frac{ED\left(S(|S|-k,|S|],S'(|S'|-k,|S'|]\right)}{2k}$$}\shortOnly{ as $l-LRSD(S,S') = \max_{0 < k < l} ED\left(S(|S|-k,|S|],S'(|S'|-k,|S'|]\right)/2k$.}
\end{definition}

Note that $l=O(\log n)$-limited suffix distance of two strings can be computed in $O(l^2)=O(\log^2 n)$ by computing edit distance of all pairs of prefixes of their $l$-long suffixes.

\begin{lemma}\label{lem:LRSD}
If string $S$ is a $c$-long distance $\eps$-synchronization string, then for any two distinct prefixes $S[1, i]$ and $S[1, j]$, $(c\log n)$-$LRSD(S[1, i], S[1, j]) > 1-\eps$.
\end{lemma}
\begin{proof}
If $j - i < c\log n$, the synchronization string property gives that 
$ED(S(2i-j, i], S(i, j]) > 2(j-i)(1-\eps)$ which gives the claim for $k = j-i$. If $j - i \ge c\log n$, the long-distance property gives that 
$ED(S(i-\log n, i], S(j-\log n, j]) > 2(1-\eps)c\log n$
which again, proves the claim.
\end{proof}

Lemmas~\ref{lem:fastDecodingPhaseOne} and \ref{lem:LRSD} enable us to prove Theorem~\ref{thm:LocalDecodingFiniteModular}.

\begin{proof}[Proof of Theorem~\ref{thm:LocalDecodingFiniteModular}]
Using Lemma~\ref{lem:fastDecodingPhaseOne}, by decoding $4/\eps$ codewords, Bob forms a list of $4N/\eps$ candidates for the index of the received symbol. This will take $4/\eps \cdot T_{Dec}(N)$ time. Then, using Lemma~\ref{lem:LRSD}, for any of the $4N/\eps$ candidates, he has to construct a $c\log n$ substring of $R$ and compute the $(c\log n)$-LRSD of that with the string he received. This requires looking at the last $\max(4n/\eps, c\log n)$ recieved symbols and takes $4N/\eps\cdot(T_{Enc}(N) + Ex_T(c\log n) + c^2\log^2n)$ time.
\end{proof}

\fullOnly{
\section{Application: Near Linear Time Codes Against Insdels, Block Transpositions, and Block Replications}\label{sec:AppilcationsCode}
In Sections~\ref{sec:sync_construction} and~\ref{sec:localDecoding}, we provided highly explicit constructions and local decodings for synchronization strings. Utilizing these two important properties of synchronization strings together suggests important improvements over insertion-deletion codes introduced by Haeupler and Shahrasbi~\cite{haeupler2017synchronization}.
We start by stating the following important lemma which summarizes the results of Sections~\ref{sec:sync_construction} and~\ref{sec:localDecoding}.

\begin{lemma}\label{thm:FastEncFastDec}
For any $0<\eps<1$, there exists an streaming $(n, \delta)$-indexing solution with $\eps$-synchronization string $S$ and streaming decoding algorithm $\mathcal{D}$ that figures out the index of each symbol by merely considering the last $O_\eps(\log n)$ received symbols and in $O_\eps(\log^3 n)$ time. Further, $S \in \Sigma^n$ is highly-explicit and constructible in linear-time and $|\Sigma| = O\left(\eps^{-O(1)}\right)$. This solution may contain up to $\frac{n\delta}{1-\eps}$ misdecodings. 
\end{lemma}
\begin{proof}
Let $S$ be a long-distance $2\eps$-synchronization string constructed according to Theorem~\ref{thm:FiniteLongDistConstructionNEW} and enhanced as suggested in~\eqref{eq:LocalDecConstruction} to ensure local decodablity. As discussed in Sections~\ref{sec:sync_construction} and \ref{sec:localDecoding}, these strings trivially satisfy all properties claimed in the statement other than the misdecoding guarantee.

According to Theorem~\ref{thm:LocalDecodingFinite}, correct decoding is ensured whenever relative suffix error density is less than $1-\frac{2\eps}{2} = 1-\eps$. Therefore, as relative suffix error density can exceed $1-\eps$ upon arrival of at most $\frac{n\delta}{1-\eps}$ many symbols (see Lemma 5.14 from~\cite{haeupler2017synchronization}), there can be at most $\frac{n\delta}{1-\eps}$ many successfully received symbols which are not decoded correctly. This proves the misdecoding guarantee.
\end{proof}

\subsection{Near-Linear Time Insertion-Deletion Code}
Using the indexing technique proposed by Haeupler and Shahrasbi~\cite{haeupler2017synchronization} summarized in Theorem~\ref{thm:mainECC} with synchronization strings and decoding algorithm from Theorem~\ref{thm:globalDecoding}, one can obtain the following insdel codes.

\begin{theorem}\label{thm:NearLinearInsdel}
For any $0<\delta<1/3$ and sufficiently small $\eps>0$, there exists an encoding map $E : \Sigma^k \rightarrow \Sigma^n$ and a decoding map $D : \Sigma^*\rightarrow \Sigma^k$, such that, if $EditDistance(E(m), x) \le \delta n$ then $D(x) = m$. Further, $\frac{k}{n}> 1-3\delta-\eps$, $|\Sigma|= f(\eps)$, and $E$ and $D$ can be computed in $O(n)$ and $O(n\log^3 n)$ time respectively.
\end{theorem}
\begin{proof}
We closely follow the proof of Theorem 1.1 from \cite{haeupler2017synchronization} and use Theorem~\ref{thm:mainECC} to convert a near-MDS error correcting code to an insertion-deletion code satisfying the claimed properties.

Given the $\delta$ and $\eps$, we choose $\eps' = \frac{\eps}{12}$ and use locally decodable $O_{\eps'}(1)$-long-distance $\eps'$-synchronization string $S$ of length $n$ over alphabet $\Sigma_S$ of size $\eps'^{-O(1)}=\eps^{-O(1)}$ from Theorem~\ref{thm:LocalDecodingFinite}.
 
We plug this synchronization string with the local decoding from Theorem~\ref{thm:LocalDecodingFinite} into Theorem~\ref{thm:mainECC} with a near-MDS expander code~\cite{guruswami2005linear} $\mathcal{C}$ (see Theorem~\ref{thm:GuruswamiIndykHighRateCodes}) which can efficiently correct up to $\delta_{\mathcal{C}} = 3\delta + \frac{\eps}{3}$ half-errors and has a rate of $R_{\mathcal{C}} > 1 - \delta_{\mathcal{C}} - \frac{\eps}{3}$ over an alphabet $\Sigma_{\mathcal{C}} = \exp(\eps^{-O(1)})$ such that $\log |\Sigma_{\mathcal{C}}|\geq \frac{3 \log |\Sigma_S|}{\eps}$. This ensures that the final rate is indeed at least $\frac{R_{\mathcal{C}}}{1+\frac{\log \Sigma_S}{\log \Sigma_{\mathcal{C}}}}\ge R_{\mathcal{C}} - \frac{\log \Sigma_S}{\log \Sigma_{\mathcal{C}}} = 1 - 3\delta - 3 \frac{\eps}{3} = 1-3\delta-\eps$ and the fraction of insdel errors that can be efficiently corrected is 
$\delta_{\mathcal{C}} - 2 \frac{\delta}{1-\eps'} \geq 
3\delta+\eps/3 - 2\delta(1+2\eps') \geq
\delta$. 
The encoding and decoding complexities are furthermore straight forward according to guarantees stated in Theorem~\ref{thm:FastEncFastDec} and the linear time construction of $S$.
\end{proof}

\subsection{Insdels, Block Transpositions, and Block Replications}\label{sec:BlockTranspositionAndReplication}
In this section, we introduce block transposition and block replication errors and show that code from Theorem~\ref{thm:NearLinearInsdel} can overcome these types errors as well. 

One can think of several way to model transpositions and replications of blocks of data. One possible model would be to have the string of data split into blocks of length $l$ and then define transpositions and replications over those fixed blocks. In other words, for message $m_1, m_2, \cdots, m_n \in \Sigma^n$, a single transposition or replication would be defined as picking a block of length $l$ and then move or copy that blocks of data somewhere in the message.

Another (more general) model is to let adversary choose any block, i.e., substring of the message he wishes and then move or copy that block somewhere in the string. Note that in this model, a constant fraction of block replications may make the message length exponentially large in terms of initial message length. We will focus on this more general model and provide codes protecting against them running near-linear time in terms of the received block length. Such results automatically extend to the weaker model that does not lead to exponentially large corrupted messages.

We now formally define \emph{$(i, j, l)$-block transposition} as follows.
\begin{definition}[$(i, j, l)$-Block Transposition]
For a given string $M=m_1\cdots m_n$, the \emph{$(i, j, l)$-block transposition} operation for $1 \le i \le i+l \le n$ and $j\in\{1,\cdots, i-1, i+l+1, \cdots, n\}$ is defined as an operation which turns $M$ into 
\fullOnly{$$M'=m_1, \cdots, m_{i-1}, m_{i+l+1}\cdots, m_j, m_i\cdots m_{i+l}, m_{j+1}, \cdots, m_n\text{ if } j > i+l$$
or
$$M'=m_1, \cdots, m_j, m_i, \cdots, m_{i+l}, m_{j+1}, \cdots, m_{i-1}, m_{i+l+1}\cdots, m_n\text{ if } j < i$$}
\shortOnly{$M'=m_1, \cdots, m_{i-1}, m_{i+l+1}\cdots, m_j, m_i\cdots m_{i+l}, m_{j+1}, \cdots, m_n$ (if $j > i+l$) or $M'=m_1, \cdots, m_j, m_i, \cdots, m_{i+l}, m_{j+1}, \cdots, m_{i-1}, m_{i+l+1}\cdots, m_n$ (if  $j < i$)}
by removing $M[i, i+l]$ and inserting it right after $M[j]$.
\end{definition}

Also, \emph{$(i, j, l)$-block replication} is defined as follows.
\begin{definition}[$(i, j, l)$-Block Replication]
For a given string $M=m_1\cdots m_n$, the \emph{$(i, j, l)$-block replication} operation for $1 \le i \le i+l \le n$ and $j\in\{1,\cdots, n\}$ is defined as an operation which turns $M$ into 
$M'=m_1, \cdots, m_j, m_i\cdots m_{i+l}, m_{j+1}, \cdots, m_n$
 which is obtained by copying $M[i, i+l]$ right after $M[j]$.
\end{definition}

We now proceed to the following theorem that implies the code from Theorem~\ref{thm:NearLinearInsdel} recovers from block transpositions and replications as well.

\begin{theorem}\label{thm:transpositionCode}
Let $S\in\Sigma_S^n$ be a locally-decodable highly-explicit $c$-long-distance $\eps$-synchronization string from Theorem~\ref{thm:LocalDecodingFinite} and $\mathcal{C}$ be an half-error correcting code of block length $n$, alphabet $\Sigma_\mathcal{C}$, rate $r$, and distance $d$ with encoding function $\mathcal{E}_\mathcal{C}$ and decoding function $\mathcal{D}_\mathcal{C}$ that run in $T_{\mathcal{E}_\mathcal{C}}$ and $T_{\mathcal{D}_\mathcal{C}}$ respectively.
Then, one can obtain an encoding function $E_n:\Sigma_\mathcal{C}^{nr} \rightarrow \left[\Sigma_\mathcal{C}\times \Sigma_S\right]^n$ that runs in $T_{\mathcal{E}_\mathcal{C}} + O(n)$ and decoding function $D_n:\left[\Sigma_\mathcal{C}\times \Sigma_S\right]^* \rightarrow \Sigma_\mathcal{C}^{nr}$ which runs in $T_{\mathcal{D}_\mathcal{C}} + O\left(\log^3 n\right)$ and recovers from $n\delta_{insdel}$ fraction of synchronization errors and $\delta_{block}$ fraction of block transpositions or replications as long as $\left(2+\frac{2}{1-\eps/2}\right)\delta_{insdel} + (12c\log n) \delta_{block} < d$. 
\end{theorem}
\begin{proof}
To obtain such codes, we simply index the symbols of the given error correcting code with the symbols of the given synchronization strings. More formally, the encoding function $\mathcal{E}(x)$ for $x\in \Sigma_\mathcal{C}^{nr}$ first computes $\mathcal{E}_\mathcal{C}(x)$ and then indexes it, symbol by symbol, with the elements of the given synchronization string. 

On the decoding end, $\mathcal{D}(x)$ first uses the indices on each symbol to guess the actual position of the symbols using the local decoding of the $c$-long-distance $\eps$-synchronization string. Rearranging the received symbols in accordance to the guessed indices, the receiving end obtains a version of $\mathcal{E}_\mathcal{C}(x)$, denoted by $\bar{x}$, that may suffer from a number of symbol corruption errors due to incorrect index misdecodings. 
As long as the number of such misdecodings, $k$, satisfies $n\delta_{insdel}+2k\le nd$, computing $\mathcal{D}_\mathcal{C}(\bar{x})$ gives $x$. The decoding procedure naturally consists of decoding the attached synchronization string, rearranging the indices, and running $\mathcal{D}_\mathcal{C}$ on the rearranged version. Note that if multiple symbols where detected to be located at the same position by the synchronization string decoding procedure or no symbols where detected to be at some position, the decoder can simply put a special symbol `?' there and treat it as a half-error. The decoding and encoding complexities are trivial.

In order to find the actual index of a received symbol correctly, we need the local decoding procedure to compute the index correctly. For that purpose, it suffices that no block operations cut or paste symbols within an interval of length $2c\log n$ before that index throughout the entire block transpositions/replications performed by the adversary and the relative suffix error density caused by synchronization errors for that symbol does not exceed $1-\eps/2$. As any block operation might cause three new cut/cop/paste edges and relative suffix error density is larger than $1-\eps/2$ for up to $\frac{1}{1-\eps/2}$ many symbols (according to Lemma 5.14 from~\cite{haeupler2017synchronization}), the positions of all but at most $k\le3n\delta_{block}\times 2c\log n + n\delta_{insdel}\left(1+\frac{1}{1-2\eps}\right)$ symbols will be decoded incorrectly via synchronization string decoding procedure. Hence, as long as $n\delta_{insdel}+2k \le 6\delta_{block}\times 2c\log n + n\delta_{insdel}\left(3+\frac{2}{1-2\eps}\right) < d$ the decoding procedure succeeds. Finally, the encoding and decoding complexities follow from the fact that indexing codewords of length $n$ takes linear time and the local decoding of synchronization strings takes $O(n\log^3 n)$ time.
\end{proof}

Employing locally-decodable $O_\eps(1)$-long-distance synchronization strings of Theorem~\ref{thm:LocalDecodingFinite} and error correcting code of Theorem~\ref{thm:GuruswamiIndykHighRateCodes} in Theorem~\ref{thm:transpositionCode} gives the following code.
\begin{theorem}
For any $0<r<1$ and sufficiently small $\eps$ there exists a code with rate $r$ that corrects $n\delta_{insdel}$ synchronization errors and $n\delta_{block}$ block transpositions or replications as long as $6\delta_{insdel} + c\log n \delta_{block} < 1-r-\eps$ for some $c=O(1)$. The code is over an alphabet of size $O_\eps(1)$ and has $O(n)$ encoding and $O(N\log^3 n)$ decoding complexities where $N$ is the length of the received message.
\end{theorem}
}
\fullOnly{\section{Applications: Near-Linear Time Infinite Channel Simulations with Optimal Memory Consumption}\label{sec:InfiniteSimulation}

We now show that the indexing algorithm introduced in Theorem~\ref{thm:FastEncFastDec} can improve the efficiency of channel simulations from~\cite{haeupler2017synchronization2:ARXIV} as well as insdel codes.
Consider a scenario where two parties are maintaining a communication that suffers from synchronization errors, i.e, insertions and deletions. Haeupler et al.~\cite{haeupler2017synchronization2:ARXIV} provided a simple technique to overcome this desynchronization. Their solution consists of a simple symbol by symbol attachment of a synchronization string to any transmitted symbol. The attached indices enables the receiver to correctly detect indices of most of the symbols he receives. However, the decoding procedure introduced in Haeupler et al.~\cite{haeupler2017synchronization2:ARXIV} takes polynomial time in terms of the communication length. The explicit construction introduced in Section~\ref{sec:sync_construction} and local decoding provided in Section~\ref{sec:localDecoding} can reduce the construction and decoding time and space complexities to polylogarithmic. Further, the decoding procedure only requires to look up $O_\eps(\log n)$ recently received symbols upon arrival of any symbol. 

Interestingly, we will show that, beyond the time and space complexity improvements over simulations in~\cite{haeupler2017synchronization2:ARXIV}, long-distance synchronization strings can make \emph{infinite channel simulations} possible. In other words, two parties communicating over an insertion-deletion channel are able to simulate a corruption channel on top of the given channel even if they are not aware of the length of the communication before it ends with similar guarantees as of~\cite{haeupler2017synchronization2:ARXIV}.
To this end, we introduce infinite strings that can be used to index communications to convert synchronization errors into symbol corruptions. 
The following theorem analogous to the indexing algorithm of Lemma~\ref{thm:FastEncFastDec} provides all we need to perform such simulations.

\begin{theorem}\label{thm:InfiniteFastEncFastDec}
For any $0<\eps<1$, there exists an infinite string $S$ that satisfies the following properties:
\begin{enumerate}
\item String $S$ is over an alphabet of size $\eps^{-O(1)}$.
\item String $S$ has a highly-explicit construction and, for any $i$, $S[i, i+\log i]$ can be computed in $O(\log i)$.
\item Assume that $S[1, i]$ is sent over an insertion-deletion channel. There exists a decoding algorithm for the receiving side that, if relative suffix error density is smaller than $1-\eps$, can correctly find $i$ by looking at the last $O(\log i)$ and knowing the number of received symbols in $O(\log^3 i)$ time.
\end{enumerate}
\end{theorem}
\begin{proof}
To construct such a string $S$, we use our finite-length highly-explicit locally-decodable long-distance synchronization string constructions from Theorem~\ref{thm:LocalDecodingFinite} and use to construct finite substrings of $S$ as proposed in the infinite string construction of Theorem~\ref{thm:infiniteSync} which is depicted in Figure~\ref{fig:constructive}. We choose length progression parameter $k=10/\eps^2$. Similar to the proof of Lemma~\ref{lem:covering}, we define \emph{turning point} $q_i$ as the index at which $S_{k^{i+1}}$ starts. Wee append one extra bit to each symbol $S[i]$ which is zero if $q_j\le i<q_{j+1}$ for some even $j$ and one otherwise.

This construction clearly satisfies the first two properties claimed in the theorem statement. To prove the third property, suppose that $S[1,i]$ is sent and received as $S'[1, i']$ and the error suffix density is less than $1-\eps$. As error suffix density is smaller than $1-\eps$, $i\eps \le i' \le i/\eps$ which implies that $i'\eps \le i \le i'/\eps$. This gives an uncertainty interval whose ends are close by a factor of $1/\eps^2$. By the choice of $k$, this interval contains at most one turning point. Therefore, using the extra appended bit, receiver can figure out index $j$ for which $q_j\le i<q_{j+1}$. Knowing this, it can simply use the local decoding algorithm for finite string $S_{j-1}$ to find $i$.
\end{proof}

\begin{theorem}\label{thm:nearLinearChannelSimulationsForInsdel}
\begin{enumerate}
\item[(a)] Suppose that $n$ rounds of a one-way/interactive insertion-deletion channel over an alphabet $\Sigma$ with a $\delta$ fraction of insertions and deletions are given. Using an $\eps$-synchronization string over an alphabet $\Sigma_{syn}$, it is possible to simulate $n\left(1-O_\eps(\delta)\right)$ rounds of a one-way/interactive corruption channel over $\Sigma_{sim}$ with at most $O_\eps\left(n\delta\right)$ symbols corrupted so long as $|\Sigma_{sim}| \times |\Sigma_{syn}| \le |\Sigma|$. 
\item[(b)] Suppose that $n$ rounds of a binary one-way/interactive insertion-deletion channel with a $\delta$ fraction of insertions and deletions are given. It is possible to simulate 
$n(1-\Theta( \sqrt{\delta\log(1/\delta)}))$
 rounds of a binary one-way/interactive corruption channel 
 with $\Theta(\sqrt{\delta\log(1/\delta)})$ fraction of corruption errors between two parties over the given channel.
\end{enumerate}
Having an explicitly-constructible, locally-decodable, infinite string from Theorem~\ref{thm:InfiniteFastEncFastDec} utilized in the simulation, all of the simulations mentioned above take $O(\log n)$ time for sending/starting party of one-way/interactive communications. Further, on the other side, the simulation spends $O(\log^3 n)$ time upon arrival of each symbol and only looks up $O(\log n)$ many recently received symbols. Overall, these simulations take a $O(n\log^3 n)$ time and $O(\log n)$ space to run. These simulations can be performed even if parties are not aware of the communication length.
\end{theorem}
\begin{proof}
We simply replace ordinary $\eps$-synchronization strings used in all such simulations in~\cite{haeupler2017synchronization2:ARXIV} with the highly-explicit locally-decodable infinite string from Theorem~\ref{thm:InfiniteFastEncFastDec} with its corresponding local-decoding procedure instead of minimum RSD decoding procedure that is used in~\cite{haeupler2017synchronization2:ARXIV}. This keeps all properties that simulations proposed by Haeupler et. al.~\cite{haeupler2017synchronization2:ARXIV} guarantee.
Further, by properties stated in Theorem~\ref{thm:InfiniteFastEncFastDec}, the simulation is performed in near-linear time, i.e., $O(n\log^3 n)$. Also, constructing and decoding each symbol of the string from Theorem~\ref{thm:InfiniteFastEncFastDec} only takes $O(\log n)$ space which leads to an $O(\log n)$ memory requirement on each side of the simulation.
\end{proof}

\exclude{\color{red}\begin{theorem}
Suppose that Alice and Bob are communicating over a $n$-round one-way insertion-deletion channel with $\delta$ fraction of errors. Indexing the symbols of communication with a locally decodable long-distance $\eps$-synchronization string $S\in\Sigma^n$ introduced in \eqref{}, enables the receiver to guess the index of every symbol it arrives that satisfies the following properties:
\begin{itemize}
\item If suffix error density is smaller than $1-2\eps$, the index is decoded correctly, therefore, this indexing solution contains $??$ misdecodings. 
\item The decoding procedure satisfies the streaming property. Hence, it leads to simulation of $??$ rounds of a corruption channel with $??$ fraction of corruptions.
\item $S$ has a highly explicit construction. 
\item Decoding procedure for every arrived symbol on receiver side takes $O_\eps(\text{poly}\log n)$ time.
\item $|\Sigma| = O_\eps(1)$.
\end{itemize}
\end{theorem}}

\section{Applications: Near-Linear Time Coding Scheme for Interactive Communication}\label{sec:NearLinearInteractiveScheme}

Using the near-linear time interactive channel simulation in Theorem~\ref{thm:nearLinearChannelSimulationsForInsdel} with the near-linear time interactive coding scheme of Haeupler and Ghaffari~\cite{ghaffari2014optimal} (stated in Theorem~\ref{GhaffariHaeupler2014Scheme}) gives the near-linear time coding scheme for interactive communication over insertion-deletion channels stated in Theorem~\ref{thm:nearLinearLargeAlphaInteractiveCodingScheme}. 

\begin{theorem}[Theorem 1.1 from \cite{ghaffari2014optimal}]\label{GhaffariHaeupler2014Scheme}
For any constant $\eps > 0$ and $n$-round protocol $\Pi$ there is a randomized non-adaptive coding scheme that robustly simulates $\Pi$ against an adversarial error rate of $\rho \le 1/4 - \eps$ using $N = O(n)$ rounds, a near-linear $n\log^{O(1)} n$ computational complexity, and failure probability $2^{-\Theta(n)}$.
\end{theorem}

\begin{theorem}\label{thm:nearLinearLargeAlphaInteractiveCodingScheme}
For a sufficiently small $\delta$ and $n$-round alternating protocol $\Pi$, there is a randomized coding scheme simulating $\Pi$ in presence of $\delta$ fraction of edit-corruptions with constant rate (i.e., in $O(n)$ rounds) and in near-linear time. This coding scheme works with probability $1-2^{\Theta(n)}$. 
\end{theorem}}

\newpage
\bibliographystyle{plain}
\bibliography{bibliography,refs}

\newpage
\newpage
\begin{center}
\bfseries \huge Appendices
\end{center}
\appendix
\shortOnly{
\PrelimsSynchStrings 
\section{Missing Proofs and Discussions from Section~\ref{sec:sync_construction}}
\ProofOfLemLongdistancereduction
\ProofOfLLLConstruction
\ProofOfCorShortDistanceSufficient
\ProofOfLemmaSimplepolyboosting
\ProofOfThmLinearTimeConstruction
\ProofOfLemCodeBlock
\ProofOfTheoremInsdelCode
\ProofOfThmFiniteLongDistConstructionNEW
\ProofOfThmInfiniteSync

}
\section{Alphabet Size vs Distance Function}\label{app:alphabetSizeVSDistance}
In this section, we study the dependence of alphabet size over distance function, $f$, for $f(l)$-distance synchronization strings. We will discuss this dependence for polynomial, exponential, and super exponential function $f$. As briefly mentioned in Section~\ref{sec:LongDistSync}, we will show that for any polynomial function $f$, one can find arbitrarily long $f(l)$-distance $\eps$-synchronization strings over an alphabet that is polynomially large in terms of $\eps^{-1}$ (Theorem~\ref{thm:LLLDistanceFunctions}). Also, in Theorem~\ref{thm:PolyDistanceLowerBoundForAlphabetSize}, we will show that one cannot hope for such guarantee over alphabets with sub-polynomial size in terms of $\eps^{-1}$. Further, for exponential distance function $f$, we will show that arbitrarily long $f(l)$-distance $\eps$-synchronization strings exist over alphabets that are exponentially large in terms of $\eps^{-1}$ (Theorem~\ref{thm:LLLDistanceFunctions}) and, furthermore, cannot hope for such strings over alphabets with sub-exponential size in terms of $\eps^{-1}$ (Theorem~\ref{thm:ExpDistanceLowerBoundForAlphabetSize}). Finally, in Theorem~\ref{thm:superExpAlphabetSize}, we will show that for super-exponential $f$, $f(l)$-distance $\eps$-synchronization string does not exist over constant-sized alphabets in terms of string length.

\begin{theorem}\label{thm:LLLDistanceFunctions}
For any polynomial function $f$, there exists an alphabet of size $O(\eps^{-4})$ over which arbitrarily long $f(l)$-distance $\eps$-synchronization strings exist. Further, for any exponential function $f$, such strings exist over an alphabet of size $\exp(\eps^{-1})$.
\end{theorem}
\begin{proof}
To prove this we follow the same LLL argument as in Theorem~\ref{thm:polyConstruction} and~\cite{haeupler2017synchronization} to prove the existence of a string that satisfies the $f(l)$-distance $\eps$-synchronization string property for intervals of length $t$ or more and then concatenate it with $1,2,\cdots, t,1,2,\cdots, t, \cdots$ to take care of short intervals. We define bad events $B_{i_1, l_1, i_2, l_2}$ in the same manner as in Theorem~\ref{thm:polyConstruction} and follow similar steps up until~\eqref{eq:LLLMiddleStep} by proposing $x_{i_1, l_1, i_2, l_2} = D^{-\eps (l_1 + l_2)}$ for some $D> 1$ to be determined later. $D$ has to be chosen such that for any $i_1, l_1, i_2, l_2$ and $l=l_1 + l_2$:
\begin{eqnarray}
\left(\frac{e}{\eps  \sqrt{|\Sigma|}}\right)^{\eps l} &\le& D^{-\eps l} \prod_
{[S[i_1, i_1+l_1)\cup S[i_2, i_2+l_2)]\cap[S[i'_1, i'_1+l'_1)\cup S[i'_2, i'_2+l'_2)]\neq \emptyset}
 \left(1-D^{-\eps (l'_1 + l'_2)}\right)
\end{eqnarray}

Note that:
\begin{eqnarray}
&&D^{-\eps l} \prod_{[S[i_1, i_1+l_1)\cup S[i_2, i_2+l_2)]\cap[S[i'_1, i'_1+l'_1)\cup S[i'_2, i'_2+l'_2)]\neq \emptyset}
 \left(1-D^{-\eps (l'_1 + l'_2)}\right)\allowdisplaybreaks \\
&\ge& D^{-\eps l} \prod_{l'=t}^{n}\prod_{l'_1=1}^{l'}
\left(1-D^{-\eps l'}\right)^{\left[(l_1+l'_1)+(l_1+l'_2)+(l_2+l'_1)+(l_2+l'_2)\right] f(l')}\allowdisplaybreaks \\
 &=& D^{-\eps l} \prod_{l'=t}^{n}
 \left(1-D^{-\eps l'}\right)^{4l'(l+l') f(l')}
 \allowdisplaybreaks \\&=& 
 D^{-\eps l}\left[\prod_{l'=t}^{n} \left(1-D^{-\eps l'}\right)^{4l'f(l')}\right]^{l}
  \times \prod_{l'=t}^{n} \left(1-D^{-\eps l'}\right)^{4{l'}^2f(l')}\allowdisplaybreaks \\
 &\ge& D^{-\eps l}\left[1-\sum_{l'=t}^{n} 4l'f(l')D^{-\eps l'}\right]^{l}
 \times  \left(1-\sum_{l'=t}^{n}4{l'}^2f(l') D^{-\eps l'}\right)
\end{eqnarray}
To bound below this term we use an upper-bound for series $\Sigma_{i=t}^{\infty} g(i) x^i$. Note that the proportion of two consecutive terms in such summation is at most $\frac{g(t+1)x^{t+1}}{g(t)x^t}$. Therefore,
$\Sigma_{i=t}^{\infty} g(i) x^i \le \frac{g(t) x^t}{1-
\frac{g(t+1)x^{t+1}}{g(t)x^t}}$.
Therefore, for LLL to work, it suffices to have the following.
\begin{eqnarray}
\left(\frac{e}{\eps  \sqrt{|\Sigma|}}\right)^{\eps l} &\le& 
D^{-\eps l}\left[1-
\frac{4tf(t) D^{-\eps t}}{1-\frac{4tf(t+1)D^{-\eps (t+1)}}{4tf(t)D^{-\eps t}}}
\right]^{l}
 \times  \left(1-
 \frac{4t^2f(t) D^{-\eps t}}{1-\frac{4t^2f(t+1)D^{-\eps (t+1)}}{4t^2f(t)D^{-\eps t}}}
 \right)\\
 &=&
D^{-\eps l}\left[1-
\frac{4tf(t) D^{-\eps t}}{1-\frac{f(t+1)D^{-\eps}}{f(t)}}
\right]^{l}
 \times  \left(1-
 \frac{4t^2f(t) D^{-\eps t}}{1-\frac{f(t+1)D^{-\eps}}{f(t)}}
 \right)\label{eqn:LLLModular}
\end{eqnarray}
 
 \noindent\textbf{Polynomial Distance Function:} For polynomial function $f(l)=\sum_{i=0}^d a_i l^i$ of degree $d$, we choose $t=1/\eps^2$ and $D=e$. This choice gives that
  $$ L_1=\frac{4tf(t) D^{-\eps t}}{1-\frac{f(t+1)D^{-\eps}}{f(t)}} = 
  \frac{4\eps^{-2}f(\eps^{-2}) e^{-1/\eps}}{1-(1+\eps^2)^d e^{-\eps}}
 $$
 and
  $$ L_2=\frac{4t^2f(t) D^{-\eps t}}{1-\frac{f(t+1)D^{-\eps}}{f(t)}} = 
  \frac{4\eps^{-4}f(\eps^{-2}) e^{-1/\eps}}{1-(1+\eps^2)^d e^{-\eps}} .$$

We study the following terms in $\eps\rightarrow 0$ regime.
Note that $4\eps^{-2}$ and $4\eps^{-4}$ are polynomials in $\eps^{-1}$ but $e^{-1/\eps}$ is exponential in $\eps^{-1}$. Therefore, for sufficiently small $\eps$,
$$4\eps^{-2}f(\eps^{-2}) e^{-1/\eps}, 4\eps^{-4}f(\eps^{-2}) e^{-1/\eps} \le
e^{-0.9/\eps}.$$
Also, $1-(1+\eps^2)^d e^{-\eps} \le 1-(1+\eps^2)^d(1-\eps/2) = 1-(1-\eps/2+o(\eps^2))$. So, for small enough $\eps$, $1-(1+\eps^2)^d e^{-\eps} \le \frac{3}{4}\eps$. This gives that, for small enough $\eps$, 
\begin{equation}
L_1, L_2 \le \frac{e^{-0.9/\eps}}{(3/4)\eps}\le e^{-0.8/\eps}.\label{eqn:L1L2}
\end{equation}
Note that $1-e^{-0.8/\eps} \ge e^{-\eps}$ for $0<\eps<1$. Plugging this fact into~\eqref{eqn:LLLModular} gives that, for small enough $\eps$, the LLL condition is satisfied if 
$$\left(\frac{e}{\eps  \sqrt{|\Sigma|}}\right)^{\eps l} \le
e^{-\eps l} \cdot e^{-\eps l}\cdot e^{-\eps }
\Leftrightarrow 
\left(\frac{e^3}{\eps\sqrt{|\Sigma|}}\right)^{\eps l}\le \frac{1}{e^\eps} 
\Leftrightarrow 
|\Sigma| \ge \frac{e^{6+2/l}}{\eps^2}
\Leftarrow
|\Sigma| \ge \frac{e^8}{\eps^2}=O(\eps^{-2})
$$
Therefore, for any polynomial $f$, $f(l)$-distance $\eps$-synchronization strings exist over alphabets of size $t\times |\Sigma| = O(\eps^{-4})$.

\bigskip

 \noindent\textbf{Exponential Distance Function:} For exponential function $f(l)=c^l$, we choose $t=1$ and $D=(8c)^{1/\eps}$. Plugging this choice of $t$ and $D$ into~\eqref{eqn:LLLModular} turns it into the following.
 \begin{eqnarray}
\left(\frac{e}{\eps  \sqrt{|\Sigma|}}\right)^{\eps l} &\le& 
D^{-\eps l}\left[1-
\frac{4tf(t) D^{-\eps t}}{1-\frac{f(t+1)D^{-\eps}}{f(t)}}
\right]^{l}
 \times  \left(1-
 \frac{4t^2f(t) D^{-\eps t}}{1-\frac{f(t+1)D^{-\eps}}{f(t)}}
 \right)\\
 &=&
(2c)^{-l}\left[1-
\frac{4c (8c)^{-1}}{1-c  \frac{1}{8c}}
\right]^{l}
 \times  \left(1-
\frac{4c (8c)^{-1}}{1-c  \frac{1}{8c}}
 \right)\\
 &=& \frac{1}{(2c)^l} \cdot \left[1-\frac{1/2}{7/8} \right]^{l+1}  = \frac{2\cdot (3/14)^{l+1}}{c^l}
\end{eqnarray}
Therefore, if $|\Sigma|$ satisfies the following, the LLL condition will be satisfied.
$$\left(\frac{e}{\eps  \sqrt{|\Sigma|}}\right)^{\eps l} \le \frac{2\cdot (3/14)^{l+1}}{c^l}
\Leftarrow
|\Sigma| \ge \frac{e^2}{\eps^2} \cdot \left(\frac{14^2 c}{3^2}\right)^{2/\eps}
$$
Therefore, for any exponential $f$, $f(l)$-distance $\eps$-synchronization strings exist over alphabets of size $c_0^{1/\eps}$ where $c_0$ is a constant depending on the basis of the exponential function $f$.
\end{proof}

\begin{theorem}\label{thm:PolyDistanceLowerBoundForAlphabetSize}
Any alphabet $\Sigma$ over which arbitrarily long $f(l)$-distance $\eps$-synchronization strings exist has to be of size $\Omega(\eps^{-1})$. This holds for any function $f$.
\end{theorem}
\begin{proof}
We simply prove this theorem for $f(l)=0$, i.e., ordinary synchronization strings which trivially extends to general $f$. Note that $\eps$-synchronization guarantee for any pair of intervals $[i,j)$ and $[j, k)$ where $k-i < \eps^{-1}$ dictates that no symbol have to appear more than once in $[i, k)$. Therefore, the alphabet size has to be at least $\eps^{-1}-1$.
\end{proof}

\begin{theorem}\label{thm:ExpDistanceLowerBoundForAlphabetSize}
Let $f$ be an exponential function. If arbitrarily long $f(l)$-distance $\eps$-synchronization strings exist over an alphabet $\Sigma$, the size of $\Sigma$ has to be at least exponentially large in terms of $\eps^{-1}$.
\end{theorem}
\begin{proof}
Let $f(l) = c^l$.
In a given $f(l)$-distance $\eps$-synchronization string, take two intervals of length $l_1$ and $l_2$ where $l_1+l_2 \le \eps^{-1}/2 < \eps^{-1}$. The edit distance requirement of $\eps$-synchronization definition requires those two intervals not to contain any similar symbols. Note that this holds for any two intervals of total length $l=\eps^{-1}/2$ in a prefix of length $c^l=c^{\eps^{-1}/2}$. Therefore, no symbol can be appear more than once throughout the first $c^{\eps^{-1}/2}$ symbols of the given strings. This shows that the alphabet size has to be at least exponentially large in terms of $\eps^{-1}$.
\end{proof}

\begin{theorem}\label{thm:superExpAlphabetSize}
For any super-exponential function $f$ and any finite alphabet $\Sigma$, there exists a positive integer $n$ such that there are no $f(l)$-distance $\eps$-synchronization strings of length $n$ or more over $\Sigma$.
\end{theorem}
\begin{proof}
Consider a substring of length $l$ in a given string over alphabet $\Sigma$. There are $|\Sigma|^l$ many possible assignments for such substring. Since $f$ is a super-exponential function, for sufficiently large $l\ge \eps^{-1}$, $\frac{f(l)}{l} \ge |\Sigma|^l$. For such $l$, consider a string of length $n\ge f(l)$. Split the first $f(l)$ elements into $\frac{f(l)}{l}$ blocks of length $l$. As $\frac{f(l)}{l} > |\Sigma|^l$, two of those blocks have to be identical. As $l$ was assumed to be larger than $\eps^{-1}$, this violates $f(l)$-distance $\eps$-synchronization property for those two blocks and therefore finishes the proof.
\end{proof}

\SecInfiniteLongDistConstruction

\end{document}